%% file: main.tex
\documentclass{eptcs}
\usepackage{etex}

\newcommand{\trefeq}[1]{(\textrm{#1})}

\newcommand{\grd}[1]{[~#1~]}

\usepackage{turnstile}
\newcommand{\dataref}[5]{\ensuremath\mathbin{#3 \st \sdststile{#2\ \st\
      #5}{#1\ \st\  #4}}} 

\newcommand{\matches}[3]{\ensuremath\mathbin{#2 \stackrel{#1}{=\!\!=} #3}}
\newcommand{\bmatches}[3]{(\matches{#1}{#2}{#3})}

\usepackage[colorinlistoftodos]{todonotes}

\newcommand{\TL}[1]{\tag{#1}\label{#1}}

\newcommand{\NoteEnv}[3]{\newenvironment{#1}{\color{#3}#2 }{}}
\definecolor{brijeshcolor}{rgb}{0.8,0,0.2}
\definecolor{iancolor}{rgb}{1,0,0}
\definecolor{lindsaycolor}{cmyk}{0.2,1,0,0}
\definecolor{johncolor}{cmyk}{1,0.3,0.4,0.3}

\NoteEnv{brijesh}{Brijesh says:}{brijeshcolor}
\NoteEnv{modified}{}{blue}
\NoteEnv{ian}{Ian}{iancolor}
\NoteEnv{lindsay}{Lindsay}{lindsaycolor}
\NoteEnv{john}{John}{johncolor}

\usepackage{amsthm}
\theoremstyle{plain}
\newcounter{thm}
\newtheorem{theorem}[thm]{Theorem}
\newtheorem{lemma}[thm]{Lemma}

\theoremstyle{definition}
\newtheorem{definition}[thm]{Definition}
\newtheorem{example}[thm]{Example}

\newcommand{\NE}[1]{\underline{#1}}

\usepackage{amsmath}
\usepackage{mine}
\usepackage{timestamp}
\usepackage{oz}
\usepackage{zed-csp}
\usepackage{rcp}
\usepackage{graphicx}
\usepackage{color}
\usepackage{stmaryrd}

\newcommand{\Fin}{{\sf finite}}
\newcommand{\Inf}{{\sf infinite}}

\def\qed{\ensuremath{{}_\Box}}

\newcommand{\st}{~{\scriptscriptstyle ^\bullet}~}

\def\figrule{\rule{\columnwidth}{0.5pt}}

\def\Init{\mathop{\textsc{Init}}}

\def\abssynt{\mathop{:\joinrel:\joinrel=}}

\newcommand{\Par}{\textstyle\mathop{\|}}

\newcommand{\Empty}{{\sf empty}}
\newcommand{\Always}{\textstyle\mathord{\boxdot}}
\newcommand{\Eventually}{\rotatebox[origin=c]{45}{$\textstyle\boxdot$}}
\newcommand{\Def}{\textstyle\mathord{\boxast}}
\newcommand{\Pos}{\rotatebox[origin=c]{45}{$\textstyle\boxast$}}

\def\adjoins{\mathbin{\varpropto}}
\newcommand{\adj}[1]{\adjoins_{_{\!#1}}}
\newcommand{\adjd}[1]{\adj{\Delta}}

\def\ch{\mathbin{;}}

\def \kif{\mathop{\mathsf{if}}}
\def \kthen{\mathbin{\mathsf{then}}}
\def \kelse{\mathbin{\mathsf{else}}}

\def \bs {\backslash}

\def \dom {\mathrm{dom}}

\def \seq {\mathrm{seq}}

\DeclareMathSymbol{\Box}{\mathord}{lasy}{"32}
\DeclareMathSymbol{\Diamond}{\mathord}{lasy}{"33}
\def\prev{\mathop\varominus}
\def\next{\mathop\varoplus}

\def\equiv{=}
\def\entails{\imp}

\begin{document}

\title{Data refinement for true concurrency}

\author{Brijesh Dongol \qquad\qquad John Derrick 
\institute{Department of Computer Science, \\
  The University of Sheffield S1 4DP, UK}
\email{B.Dongol@sheffield.ac.uk, J.Derrick@dcs.shef.ac.uk}}

\def\authorrunning{B. Dongol and J.Derrick}
\def\titlerunning{Data refinement for true concurrency}
\maketitle
\begin{abstract}
  The majority of modern systems exhibit sophisticated concurrent
  behaviour, where several system components modify and observe the
  system state with fine-grained atomicity. Many systems (e.g.,
  multi-core processors, real-time controllers) also exhibit truly
  concurrent behaviour, where multiple events can occur
  simultaneously.  This paper presents data refinement defined in
  terms of an interval-based framework, which includes high-level
  operators that capture non-deterministic expression evaluation. By
  modifying the type of an interval, our theory may be specialised to
  cover data refinement of both discrete and continuous systems. We
  present an interval-based encoding of forward simulation, then prove
  that our forward simulation rule is sound with respect to our data
  refinement definition. A number of rules for decomposing forward
  simulation proofs over both sequential and parallel composition are
  developed.
\end{abstract}

\section{Introduction}

Data refinement allows one to develop systems in a stepwise manner,
enabling an abstract system to be replaced with a more concrete
implementation by guaranteeing that every observable behaviour of the
concrete system is a possible observable behaviour of the abstract.  A
benefit of such developments is the ability to reason at a level of
abstraction suitable for the current stage of development, and the
ability to introduce additional detail to a system via
correctness-preserving transformations. During development, a concrete
system's internal representation of data often differs from the
abstract data representation, requiring the use of a \emph{refinement
  relation} to link the concrete and abstract
states.  

Over the years, numerous techniques for verifying data refinement
techniques have been developed for a number of application domains
\cite{deRoever98}, including methods for refinement of concurrent
\cite{DB03} and real-time \cite{Fid93} systems. However, these
theories are rooted in traditional notions of data refinement, where
refinement relations are between concrete and abstract states. In the
presence of fine-grained atomicity and truly concurrent behaviour
(e.g., multi-core computing, real-time controllers), proofs of
refinement are limited by the information available within a single
state, and hence, reasoning can often be more difficult than
necessary. Furthermore, the behaviours of corresponding concrete and
abstract steps may not always match, and hence, reasoning can
sometimes be unintuitive, e.g., for the state-based data refinement in
\refsec{sec:background}, a concrete step that loads a variable
corresponds to an abstract step that evaluates a guard.

When reasoning about concurrent and real-time systems, one is often
required to refer to a system's evolution over time as opposed to its
current state at a single point in time. This paper therefore presents
a method for verifying data refinement using a framework that allows
one to consider the intervals within which systems execute
\cite{DDH12,DH12MPC,Mos00,ZH04}. Thus, instead of reasoning over the
pre and post states of each component, one is able to reason about the
component's behaviour over an interval, which may comprise several
atomic steps. The concurrent execution of two or more processes is
defined as the conjunction of the behaviour of each process in the
same interval \cite{AL95,Heh90}; hence, reasoning about a component
naturally takes into account the behaviour of the component's
environment (e.g., other concurrently executing processes). Using an
interval-based framework enables us to incorporate methods for
apparent states evaluation \cite{DDH12,DH12MPC,HBDJ13}, which allows
one to take into account the low-level non-determinism of expression
evaluation at a high level of abstraction.


The main contribution of this paper is an interval-based method for
verifying data refinement, simplifying data refinement proofs in the
presence of true concurrency. A forward simulation rule for
interval-based refinement is developed, and several methods of
decomposing proof obligations are presented, including mixed-mode
refinement, which enables one to establish different refinement
relations over disjoint parts of the state space. We present our
theory at the semantic level of interval predicates, i.e., without
consideration of any particular programming framework. Hence, the
theory can be applied to any existing framework such as action
systems, Z, etc. by mapping the syntactic constructs to our interval
predicate semantics. The aim of our work is to reason about programs
with fine-grained atomicity and real-time properties, as opposed to
programs written in, say, Java that allows specification of
coarse-grained atomicity using \texttt{synchronized} blocks.

Background material for the paper is
presented in Sections \ref{sec:background} and
\ref{sec:interv-based-reas}, clarifying our notions of state-based
refinement and interval-based reasoning. Our interval-based refinement
theory is presented in \refsec{sec:gener-theory-refin}, which includes
a notion of forward simulation with respect to intervals and methods
for proof decomposition. Methods for reasoning about fine-grained
concurrency and a proof of our running example is presented in
\refsec{sec:fine-grain-atom}.

\section{State-based data refinement}
\label{sec:background}

\begin{figure}[t]
  \begin{minipage}[b]{0.44\columnwidth}
    \centering
    $
    \begin{array}[t]{@{}l@{}}
      AInit: \neg grd \\
      \hline 
      \begin{array}[t]{@{}l@{}|@{}l@{}}
      \begin{array}[t]{@{}l@{}}
        \hfill \textrm{Process $ap$} \hfill \\
        \hline
        \begin{array}[t]{@{}l@{~}l@{}}
          \ \ ap_1:& {\bf if}\ grd \Then \ \  \\
          \ \ ap_2:& \ \ \ \ m \asgn 1 \\
          \ \ ap_3:& \Else m \asgn 2 \Fi\ 
        \end{array}
      \end{array}
      &
      \begin{array}[t]{@{}l@{}}
        \hfill \textrm{Process $aq$} \hfill \\
        \hline
        \begin{array}[t]{@{}l@{~}l@{}}
          \ \ aq_1:& \If b \Then \\
          \ \ aq_2:& \ \ \ \ grd \asgn true  \\
          \ \ aq_3:& \Else {\bf skip} \Fi \\
        \end{array}
      \end{array}
    \end{array}
  \end{array}
    $
    \caption{Abstract program with guard $grd$}
    \label{fig:abs-ref-atom}
  \end{minipage}
  \hfill
  \begin{minipage}[b]{0.49\columnwidth}
    \centering
    $
    \begin{array}[t]{@{}l@{}}
      CInit: v \leq u < \infty \\
      \hline 
    \begin{array}[t]{@{}l@{}|@{}l@{}}
      \begin{array}[t]{@{}l@{}}
        \hfill \textrm{Process $cp$} \hfill \\
        \hline
        \begin{array}[t]{@{}l@{~}l@{}}
          \ \ cp_1:& {\bf if}\ u < v \Then \ \  \\
          \ \ cp_2:& \ \ \ \ m \asgn 1 \\
          \ \ cp_3:& \Else m \asgn 2 \Fi
        \end{array}
      \end{array}
      &
      \begin{array}[t]{@{}l@{}}
        \hfill \textrm{Process $cq$} \hfill \\
        \hline
        \begin{array}[t]{@{}l@{~}l@{}}
          \ \ cq_1:& \If 0 < u \Then \\
          \ \ cq_2:& \ \ \ \ v \asgn \infty \\
          \ \ cq_3:& \Else v\asgn -\infty \Fi
        \end{array}
      \end{array}
    \end{array}
  \end{array}
    $
    \caption{Concrete program with guard $u < v$}
    \label{fig:conc-ref-atom}
  \end{minipage}\medskip
  
\end{figure}

Consider the abstract program in \reffig{fig:abs-ref-atom}, written in
the style of Feijen and van Gasteren \cite{FvG99}, which consists of
variables $grd, b \in \bool$, $m \in \nat$, initialisation $AInit$ and
processes $ap$ and $aq$. Process $ap$ is a sequential program with
labels $ap_1$, $ap_2$, and $ap_3$ that tests whether $grd$ holds
(atomically), then executes $m \asgn 1$ if $grd$ evaluates to $true$
and executes $m \asgn 2$ otherwise. Process $aq$ is similar. The
program executes by initialising as specified by $AInit$, and then
executing $ap$ and $aq$ concurrently by interleaving their atomic
statements.

An initialisation may be modelled by a relation, and each label
corresponds to an atomic statement, whose behaviour may also be
modelled by a relation. Thus, a program generates a set of
\emph{traces}, each of which is a sequence of states (starting with
index $0$). Program counters for each process are assumed to be
implicitly included in each state to formalise the control flow of a
program \cite{Don09}, e.g., the program in \reffig{fig:abs-ref-atom}
uses two program counters $pc_{ap}$ and $pc_{aq}$, where $pc_{ap} =
ap_1$ is assumed to hold whenever control of process $ap$ is at
$ap_1$. After execution of $ap_1$, the value of $pc_{ap}$ is updated
so that either $pc_{ap} = ap_2$ or $pc_{ap} = ap_3$ holds, depending
on the outcome of the evaluation of $grd$.

One may characterise traces using an
\emph{execution}, which is a sequence of labels starting with
initialisation.  For example, a possible execution of the program in
\reffig{fig:abs-ref-atom} is
\begin{equation}
\aang{AInit, ap_1, aq_1, aq_2, ap_3}\label{eq:8}
\end{equation}
Using `.'  for function application, an execution $ex$ corresponds to
a trace $tr$ iff for each $i \in \dom.ex$, $(tr.i, tr.(i+1)) \in ex.i$
and either $\dom.tr = \dom.ex = \nat$ or $size.(\dom.ex) <
size.(\dom.tr)$.  An execution $ex$ is \emph{valid} iff $\dom.ex \neq
\emptyset$, $ex.0$ is an initialisation, and $ex$ corresponds to at
least one trace, e.g., \refeq{eq:8} above is valid. Not every
execution is valid, e.g., $\aang{AInit, ap_1, ap_2}$ is invalid
because execution of $ap_1$ after $AInit$ causes $grd$ to evaluate to
$false$ and $pc_{ap}$ to be updated to $ap_3$, and hence, statement
$ap_2$ cannot be executed. Note that valid executions may not be
complete; an extreme
example is $\aang{AInit, AFin}$, where the execution is finalised
immediately after initialisation.

Now consider the more concrete program in \reffig{fig:conc-ref-atom}
that replaces $grd$ by $u < v$ and $b$ by $0 < u$. Note that $u$ and $v$
are fresh with respect the program in
\reffig{fig:abs-ref-atom}. Initially, $v \leq u < \infty$
holds. Furthermore, $cq$ (modelling the concrete environment of $cp$)
sets $v$ to $\infty$ if $u$ is positive and to $- \infty$
otherwise. One may be interested in knowing whether the program in
\reffig{fig:conc-ref-atom} \emph{data refines} the program in
\reffig{fig:abs-ref-atom}, which defines conditions for the program in
\reffig{fig:abs-ref-atom} to be substituted by the program in
\reffig{fig:conc-ref-atom} \cite{deRoever98}. This is possible if
every execution of the program in \reffig{fig:conc-ref-atom} has a
corresponding execution of the program in \reffig{fig:abs-ref-atom},
e.g., concrete execution $\aang{CInit, cp_1,cq_1, cq_2, cp_3}$ has a
corresponding abstract execution \refeq{eq:8}.

In general, representation of data within a concrete program differs
from the representation in the abstract, and hence, one must
distinguish between the disjoint sets of \emph{observable} and
\emph{representation} variables, which respectively denote variables
that can and cannot be observed. For example, $grd$ in
\reffig{fig:abs-ref-atom} and $u$, $v$ in \reffig{fig:conc-ref-atom}
cannot both be observable because the types of these variables are
different in the two programs. To verify data refinement, the abstract
and concrete programs may therefore also be associated with
\emph{finalisations}, which are relations between a representation and
an observable state. Different choices for the finalisation allow
different parts of the program to become observable and affect the
type of refinement that is captured by data refinement
\cite{DB03,DB07,DB09}. For the programs in Figures
\ref{fig:abs-ref-atom} and \ref{fig:conc-ref-atom}, we assume
finalisations make variable $m$ observable. Hence, Figure
\ref{fig:abs-ref-atom} is data refined by Figure
\ref{fig:conc-ref-atom} if $ap$ is able to execute $ap_2$ (and $ap_3$)
whenever $cp$ is able to execute $cp_2$ (and $cp_3$, respectively).
We define a \emph{finalised execution} of a program to be a valid
execution concatenated with the finalisation of the program, e.g.,
$\aang{AInit, ap_1, aq_1, aq_2, ap_3, AFin}$ is a finalised execution
of the program in \reffig{fig:abs-ref-atom} generated from the valid
execution \refeq{eq:8}. Valid executions are not necessarily complete,
and hence, one may observe the state in the ``middle'' of a program's
execution.

To define data refinement, we assume that an initialisation is a
relation from an observable state to a representation state, each
label corresponds to a statement that is modelled by a relation
between two representation states, and a finalisation is a relation
from a representation state to an observable state. Assuming `$\semi$'
denotes relational composition and $id$ is the identity relation, we
define the \emph{composition} of a sequence of relations $R$ as
$$
comp.R \sdef \kif R = \aang{} \kthen id \kelse head.R \semi
comp.(tail.R)
$$  
which composes the relations of $R$ in order. We also define a
function $rel$, which replaces each label in an execution by the
relation corresponding to the statement of that label.

We allow finite stuttering in the concrete program, and hence, there
may not be a one-to-one correspondence between concrete and abstract
executions. Stuttering is reflected in an abstract execution by
allowing a finite number of labels `$Id$' to be interleaved with each
finalised execution of the abstract program, where $Id$ is assumed to
be different from all other labels, and the relation corresponding to
label $Id$ is always $id$. Data refinement is therefore defined with
respect to a \emph{correspondence function} that maps concrete labels
to abstract labels. A correspondence function is valid iff it maps
concrete initialisation to abstract initialisation, concrete
finalisation to abstract finalisation, each label of a non-stuttering
concrete statement to a corresponding abstract statement, and each
label of stuttering concrete statement to $Id$. For the rest of the
paper we assume that the correspondence functions under consideration
are valid. A program $C$ is a \emph{data refinement} of a program $A$
with respect to correspondence function $f$ iff for every finalised
execution $exc$ of $C$, $exa \sdef \lambda i : \dom.exc \st f.(exc.i)$
is a finalised execution of $A$ (with possibly finite stuttering) and
$comp.(rel.exc) \subseteq comp.(rel.exa)$ holds.

Proving data refinement directly from its formal definition is
infeasible.  Instead, one proves data refinement by verifying
\emph{simulation} between an abstract and concrete system, which
requires the use of \emph{refinement relation} to link the internal
representations of the abstract and concrete programs. We assume that
a relation $r \in X \rel Y$ is characterised by a function $fr \in X
\fun Y \fun \bool$ where $(x, y) \in r$ iff $fr.x.y.$ hold. As
depicted in \reffig{fig:dref}, a refinement relation $ref$ is a
\emph{forward simulation} between a concrete and abstract system if:
\begin{enumerate}
\item whenever the concrete system can be initialised from an
  observable state $\rho$ to obtain a concrete representation state
  $\tau_0$, it must be possible to initialise the abstract system from
  $\rho$ to result in abstract representation state $\sigma_0$ such
  that $ref.\sigma_0.\tau_0$ holds,
\item for every non-stuttering concrete statement $cs$, abstract state
  $\sigma$ and concrete state $\tau$, if $ref.\sigma.\tau$ holds and
  $cs$ relates $\tau$ to $\tau'$, then there exists an abstract state
  $\sigma'$ such that the abstract statement that corresponds to $cs$
  relates $\sigma$ to $\sigma'$ and $ref.\sigma'.\tau'$ holds,

\item for every stuttering concrete statement starting from state
  $\tau$ and ending in state $\tau'$, $ref.\sigma.\tau'$ holds
  whenever $ref.\sigma.\tau$ holds,
\item finalising any abstract state $\sigma$ (using the abstract
  system's finalisation) and concrete state $\tau$ (using the concrete
  system's finalisation) results in the same observable state whenever
  $ref.\sigma.\tau$ holds.
\end{enumerate}

For models of computation that assume instantaneous guard evaluation
\cite{HBDJ13}, establishing a data refinement between the programs in
Figures \ref{fig:abs-ref-atom} and \ref{fig:conc-ref-atom} with
respect to a correspondence function that maps $cp_i$ to $ap_i$ and
$cq_i$ to $aq_i$ for $i \in \{1,2,3\}$ is straightforward. In
particular, it is possible to prove forward simulation using $pcuv$
below as the refinement relation, where $\sigma$ and $\tau$ are
abstract and concrete states, respectively.
\begin{eqnarray*}
  uv.\sigma.\tau & \sdef & (\sigma.grd = (\tau.u < \tau.v)) \land (\sigma.b
  = (0 < \tau.u)) \land (\sigma.m = \tau.m)
  \\
  pcuv.\sigma.\tau & \sdef & uv.\sigma.\tau \land \all i :
  \{1,2,3\} \st (\sigma.pc_{ap} = ap_i \imp \tau.pc_{cp} = cp_i) \land
  (\sigma.pc_{aq} = aq_i \imp \tau.pc_{cq} = cq_i) 
\end{eqnarray*}

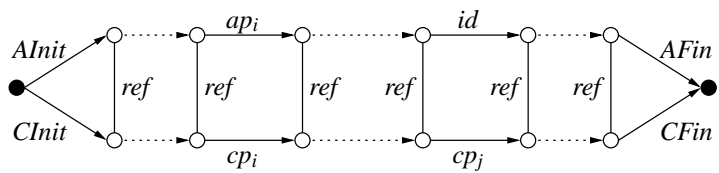
\begin{figure}[h]
    \centering
    \input{dref.pspdftex}
    \caption{Data refinement via simulation}
    \label{fig:dref}

\end{figure}

\begin{figure}[t ]
    \begin{minipage}[t]{0.49\columnwidth}
    \centering $\begin{array}[t]{@{}l@{}}
      CInit: v \leq u < \infty \\
      \hline
      \begin{array}[t]{@{}l@{}|@{}l@{}}
        \begin{array}[t]{@{}l@{}}
          \hfill \textrm{Process $cp$} \hfill \\
          \hline
          \begin{array}[t]{@{}r@{~}l@{}}
            \begin{array}[c]{@{}r@{}}
              cp_{1.1}: \\
              cp_{1.2}: 
            \end{array}
            &
            \left(\begin{array}[c]{@{}l@{}}
                k_u \asgn u \ch {}\\
                k_v \asgn v
              \end{array}\right)
            \\
            & \sqcap \\
            \ \
            \begin{array}[c]{@{}r@{}}
              cp_{1.3}: \\
              cp_{1.4}: 
            \end{array}
            &
            \left(\begin{array}[c]{@{}l@{}}
                k_v \asgn v \ch {}\\
                k_u \asgn u
              \end{array}\right)
            \ch \\
            \ \ cp_{1.5}:& {\bf if}\ k_u < k_v \Then \dots \ \ 
          \end{array}
        \end{array}
        &
        \begin{array}[t]{@{}l@{}}
          \hfill \textrm{Process $cq$} \hfill \\
          \hline
          \begin{array}[t]{@{}l@{~}l@{}}
            \ \ cq_1:& \If 0 < u \Then \\
            \ \ cq_2:& \ \ \ \ v \asgn \infty \\
            \ \ cq_3:& \Else v\asgn -\infty \Fi
          \end{array}
        \end{array}
      \end{array}
    \end{array}$
    \caption{Making the atomicity of expression evaluation in
      \reffig{fig:conc-ref-atom} explicit}
    \label{fig:split}
  \end{minipage}
  \hfill
    \begin{minipage}[t]{0.49\columnwidth}
    \centering 
    $
    \begin{array}[t]{@{}c@{~}|@{~}c@{}}
      \textrm{Concrete label} & \textrm{Abstract label}
      \\
      \hline 
      CInit & AInit \\
      cp_{1.1}, cp_{1.3}, cp_{1.5}& Id \\
      cp_{1.2}, cp_{1.4} & ap_1 \\
      cp_i \textrm{\quad for $i \in \{2,3\}$} & ap_i \\
      cq_i \textrm{\quad for $i \in \{1,2,3\}$} & aq_i \\
      CFin  & AFin 
      \\{}
    \end{array}
    $
    \caption{Correspondence function for data refinement between
      \reffig{fig:split} and \reffig{fig:abs-ref-atom}} 
    \label{fig:split-correspondence}
  \end{minipage}
\end{figure}

In a setting with fine-grained atomicity, the program in
\reffig{fig:conc-ref-atom} may be difficult to implement because the
guard at $cp_{1}$ (which refers to multiple shared variables) is
assumed to be evaluated atomically. In reality, there may be
interference from other processes while an expression is being
evaluated \cite{HBDJ13}. Furthermore, the order in which variables are
read within an expression is often not fixed.  To take these
circumstances into account, we must consider the program in
\reffig{fig:split}, which splits the guard evaluation at $cp_1$ in
\reffig{fig:conc-ref-atom} into a number of smaller atomic statements
using fresh variables $k_u$ and $k_v$ that are local to process
$cp$. Via a non-deterministic choice `$\sqcap$', process $cp$ chooses
between executions $cp_{1.1} \ch cp_{1.2}$ and $cp_{1.3} \ch
cp_{1.4}$, which read the (global values) $u$ and $v$ into local
variables $k_u$ and $k_v$, respectively, in two atomic
steps. Evaluation of guard $u < v$ at $cp_1$ in
\reffig{fig:conc-ref-atom} is then replaced by evaluation of $k_u <
k_v$.

A proof of data refinement between the programs in Figures
\ref{fig:abs-ref-atom} and \ref{fig:conc-ref-atom} using forward
simulation with respect to $uv$ is now more difficult because an
(atomic) instantaneous evaluation of $grd$ has been split into several
atomic statements.  A data refinement with respect to a naive
correspondence function that matches $cp_{i}$ for $i \in \{1.1, 1.2,
1.3, 1.4\}$ with $Id$, $cp_{1.5}$ with $ap_1$, and $cq_i$ with $aq_i$
for $i \in \{1,2\}$ cannot be verified using forward simulation.
Instead, one must use the correspondence function in
\reffig{fig:split-correspondence}. Note that this correspondence
function is not intuitive because, for example, execution of
$cp_{1.4}$ (which reads $u$) is matched with execution of $ap_1$
(which tests $grd$), but is necessary because execution of $cp_{1.4}$
determines the outcome of the future evaluation of the guard at
$cp_{1.5}$. The refinement relation used to prove forward simulation
is more complicated than $pcuv$ (details are elided, but the relation
can be constructed using the correspondence function in
\reffig{fig:split-correspondence}).

Such difficulties in verifying a relatively trivial modification
expose the complexities in stepwise refinement of concurrent
programs. Further issues arise in the context of real-time properties
e.g., transient properties cannot be properly addressed by an inherent
interleaving model \cite{DH10,DH12MPC}.

This paper presents an interval-based semantics for the systems under
consideration, an interval-based interpretation of data refinement in
the framework, and a rule akin to forward simulation for proving data
refinement. We believe that these theories alleviate many of these
issues in state-based reasoning, requiring less creativity on the part
of the verifier. For example, the correspondence function always maps
each concrete process to an abstract process. By reasoning about the
traces of a system over an interval, we are able to capture the effect
of a number of atomic statements and interference from the environment
at a high-level of abstraction. Unlike the state-based approach
described above, which only captures interleaved concurrency,
interval-based approaches also allow one to model truly concurrent
behaviour. By modifying the type of an interval, one can take both
discrete and continuous system behaviours into account.

\section{Interval-based reasoning}
\label{sec:interv-based-reas}

Our generic theory of refinement is based on interval predicates,
generalising frameworks that model programs as relations between
pre/post states. We have applied our interval-based methodology to
reason about both concurrent \cite{DD12,DDH12} and real-time programs
\cite{DH12MPC,DHMS12}.

An {\em interval} in an ordered set $\Phi \subseteq \real$ is a
contiguous subset of $\Phi$, i.e., the set of all intervals of $\Phi$
is given by:
$$
Intv_\Phi \sdef \{\Delta \subseteq \Phi \mid \all t, t' : \Delta \st
\all t'' : \Phi \st t \leq t'' \leq t' \imp t'' \in \Delta\}
$$ 
We assume the existence of elements $-\infty, \infty \notin \Phi$ such
that $-\infty < t < \infty$ for each $t \in \Phi$. 
$Intv_\Phi$ may be used to model both discrete (e.g., by picking $\Phi
= \integer$) and continuous (by picking $\Phi = \real$)
systems. 

We define the following predicates, which may be used to identify
empty intervals, and intervals with a finite and infinite upper bound.
\begin{align*}
  \Empty.\Delta\ \  \sdef\ \ & \Delta = \emptyset
  & 
  \Fin.\Delta \ \ \sdef\ \  & \Empty.\Delta \lor (\exists t : \Delta \st \all t' : 
  \Delta \st t' \leq t)
  & 
  \Inf.\Delta\ \  \sdef\ \ & \neg \Fin.\Delta 
\end{align*}
One must often reason about two \emph{adjoining} intervals, i.e.,
intervals that immediately precede/follow another. For $\Delta_1,
\Delta_2 \in Intv_\Phi$, we define
\begin{eqnarray*}
\Delta_1 \adjoins \Delta_2
& \sdef &
  \begin{array}[t]{@{}l@{}}
    (\all t_1: \Delta_1, t_2: \Delta_2 \st t_1 < t_2)
    \land
    (\Delta_1 
    \cup \Delta_2 \in Intv_\Phi)
  \end{array}
\end{eqnarray*}
Thus, $\Delta_1 \adjoins \Delta_2$ holds iff $\Delta_2$ follows
$\Delta_1$ and the union of $\Delta_1$ and $\Delta_2$ forms an
interval (i.e., $\Delta_1$ and $\Delta_2$ are contiguous across their
boundary).  Note that adjoining intervals are disjoint and that both
$\Delta \adjoins \emptyset$ and $\emptyset \adjoins \Delta$ hold
trivially for any interval $\Delta$.

A \emph{state} over $V \subseteq Var$ is of type $State_V \sdef V \fun
Val$, where $Var$ is the type of a variable and $Val$ is the generic
type of a value. A {\em state predicate} is of type $StatePred_V \sdef
State_V \fun \bool$.  A \emph{stream} of behaviours over $State_V$ is
given by the function $Stream_{\Phi,V} \sdef \Phi \fun State_V$, which
maps each element of $\Phi$ to a state over $V$. To facilitate
reasoning about specific parts of a stream, we use \emph{interval
  predicates}, which have type $IntvPred_{\Phi,V} \sdef Intv_\Phi \fun
Stream_{\Phi, V} \fun \bool$. A visualisation of an interval predicate
over $Z \subseteq Var$ is given in \reffig{fig:inv-pred-vis}. The
stream $z \in Stream_{\Phi,Z}$ maps each time to a state over $Z$ and
the interval predicate depicted in the figure maps $\Delta$ and $z$ to
a boolean.

\begin{figure}[t]
  \centering 
  \scalebox{0.8}{\input{intv-predrel.pspdftex}}
  \caption{Interval predicate visualisation}
  \label{fig:inv-pred-vis}
\end{figure}
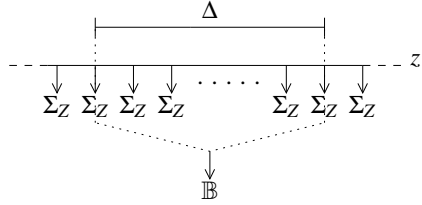

We assume pointwise lifting of operators on stream and interval
predicates in the normal manner, e.g., if $g_1$ and $g_2$ are interval
predicates, $\Delta$ is an interval and $s$ is a stream, we have $(g_1
\land g_2).\Delta.s = (g_1.\Delta.s \land g_2.\Delta.s)$.  The
\emph{chop} operator `;' is a basic operator on two interval
predicates \cite{DDH12,DH12MPC,Mos00,ZH04}, where $(g_1 \ch
g_2).\Delta$ holds iff either interval $\Delta$ may be split into two
parts so that $g_1$ holds in the first and $g_2$ holds in the second,
or the upper bound of $\Delta$ is $\infty$ and $g_1$ holds in
$\Delta$. Thus, for a stream $s$, we define:
\begin{eqnarray*}
  (g_1 \ch g_2).\Delta.s & \sdef &
  \begin{array}[t]{@{}l@{}}
    \left(\begin{array}[c]{@{}l@{}}
        \exists \Delta_1, \Delta_2 : Intv_\Phi  \st
        (\Delta = \Delta_1
        \cup \Delta_2) \land  (\Delta_1 \adjoins \Delta_2) 
        \land g_1.\Delta_1.s \land g_2.\Delta_2.s
      \end{array}\right) \lor \\
    (\Inf.\Delta \land g_1.\Delta.s)
  \end{array} 
\end{eqnarray*}
Note that $\Delta_1$ may be empty, in which case $\Delta_2 = \Delta$,
and similarly $\Delta_2$ may empty, in which case $\Delta_1 = \Delta$,
i.e., both $(\Empty \ch g) \equiv g$ and $g \equiv (g \ch \Empty)$
trivially hold, where $\Empty.\Delta.s \sdef (\Delta = \emptyset)$ for
all streams $s$. Furthermore, in the definition of chop, we allow the
second disjunct $\Inf.\Delta \land g_1.\Delta$ to enable $g_1$ to
model an infinite (divergent or non-terminating) program.

To model looping of a behaviour modelled by interval predicate $g$, we
use an iteration operator `$g^\omega$', which is defined as the
greatest fixed point of $\lambda h \st g \ch h \lor \Empty$. Interval
predicates are assumed to be ordered using implication `$\entails$'
and the greatest fixed point allows $g^\omega$ to model both finite
(including 0) and infinite iteration \cite{DHMS12}.
\begin{eqnarray*}
  g^\omega & \sdef &  \nu z \st (g \ch z) \lor
  \Empty 
\end{eqnarray*}
We say that $g$ \emph{splits} iff $g \entails (g \ch g)$ and $g$
\emph{joins} iff $(g \ch g^\omega) \entails g$. If $g$ splits, then
whenever $g$ holds in an interval $\Delta$, $g$ also holds in any
subinterval of $\Delta$. If $g$ joins, then $g$ holds in $\Delta$
whenever there is a partition of $\Delta$ such that $g$ holds in each
interval of the partition. Note that if $g$ splits, then $g \entails
g^\omega$ \cite{DHMS12}. Splits and joins properties are useful for
decomposing proof obligations, for instance, both of the following
hold.
\begin{align}
  \label{eq:18}
  (g \entails g_1) \land (g
  \entails g_2)  &\ \  \imp\ \  (g
  \entails g_1 \ch g_2)  
  & \textrm{provided $g$ splits}
  \\
  \label{eq:19}
  (g \land g_1) \ch (g \land g_2)  &\ \
  \entails\ \  g \land (g_1 \ch g_2)
  & \textrm{provided $g$ joins}
\end{align}
One must often state that a property only holds for a non-empty
interval, and that a property holds for an immediately preceding
interval. To this end, we define:
\begin{align*}
  \NE g \ \ \sdef\ \ & g \land \neg \Empty & \prev g.\Delta.s \ \
  \sdef\ \ & \exists \Delta_0 : Intv_\Phi \st \Delta_0 \adjoins \Delta
  \land g.\Delta_0.s
\end{align*}
Note that if $g$ holds in an empty interval, then $\prev g$ trivially
holds. Also note how interval predicates allow the behaviour outside
the given interval to be stated in a straightforward manner because a
stream encapsulates the entire behaviour of a system.  We define the
following operators to formalise properties over an interval using a
state predicate $c$ over an interval $\Delta$ in stream $s$.
\begin{align*}
  \Always c.\Delta.s  \ \ \sdef \ \ &  \all t : \Delta \st c.(s.t) & 
  \Eventually c.\Delta.s  \ \ \sdef\ \  &  \exists t : \Delta \st c.(s.t)
\end{align*}
That is $\Always c.\Delta.s$ holds iff $c$ holds for each state $s.t$
where $t \in \Delta$ and $\Eventually c.\Delta.s$ holds iff $c$ holds
in some state $s.t$ where $t \in \Delta$. Note that $\Always c$
trivially holds for an empty interval, but $\Eventually c$ does
not. For the rest of this paper, we assume that the underlying type of the
interval under consideration is fixed. Hence, to reduce notational
complexity, we omit $\Phi$ whenever possible.

\begin{example} 
  \label{ex:uv-ex-ip}
  We present the interval-based semantics of the programs in Figures
  \ref{fig:abs-ref-atom} and \ref{fig:conc-ref-atom}. Interval-based
  methods allow one to model true concurrency by defining the
  behaviour of a parallel composition $p \| q$ over an interval
  $\Delta$ as the conjunction of the behaviours of both $p$ and $q$
  over $\Delta$ (see \cite{DD12,DDH12,DH12MPC} for more
  details). Others have also treated parallel composition as
  conjunction, but in an interleaving framework with predicates over
  states as opposed to intervals (e.g., \cite{AL95,Heh90}). Sequential
  composition is formalised using the chop operator. We assume
  $\grd{grd}$ denotes an interval predicate that formalises evaluation
  of $grd$. Details of guard evaluation are given in
  \refsec{sec:appar-stat-eval}.  The interval-based semantics of the
  programs in Figures \ref{fig:abs-ref-atom} and
  \ref{fig:conc-ref-atom} are respectively formalised by the interval
  predicates \refeq{eq:777}, \refeq{eq:88}, \refeq{eq:55} and
  \refeq{eq:66} below. Assuming that $\rho$ is an observable state,
  conditions \refeq{eq:777} and \refeq{eq:88} formalise the behaviours
  of $AInit.\rho$ and $CInit.\rho$, respectively. Assuming that $\rho$
  has an observable variable $M$ that is represented internally by
  $m$, and that $\sigma$ and $\tau$ are abstract and concrete states,
  respectively, the behaviours of both $AFin.\sigma.\rho$ and
  $CFin.\sigma.\rho$ are formalised by \refeq{eq:10} and
  \refeq{eq:20}, respectively. We assume `;' binds more tightly than
  binary boolean operators.
  \begin{eqnarray}
    & 
    \underline{\Always \neg grd}
   \label{eq:777}
   \\
   & 
   \underline{\Always (v \leq u < \infty)}
   \label{eq:88}
   \\
   & 
   \begin{array}[b]{@{}l@{}}
        \overbrace{(\grd{grd} \ch \NE{\Always(m =  1)} \lor \grd{\neg grd} \ch
          \NE{\Always(m =  2)})}^\text{Process $ap$}
        \land 
        \overbrace{(\grd{b} \ch \underline{\Always grd} \lor \grd{\neg b})
      }^\text{Process $aq$}
      \end{array}\label{eq:55}
    \\
    & 
    \begin{array}[c]{@{}l@{}}
        \overbrace{(\grd{u < v} \ch \NE{\Always(m =  1)} \lor \grd{u \geq v} \ch
          \NE{\Always(m =  2)})}^\text{Process $cp$}
        \land 
        \overbrace{(\grd{0 < u} \ch \underline{\Always (v = \infty)}
          \lor \grd{0 \geq u} \ch
        \underline{\Always (v = -\infty)})}^\text{Process $cq$}
    \end{array}\label{eq:66}
    \\
    & \sigma.m = \rho.M 
    \label{eq:10}
    \\
    & \tau.m = \rho.M
    \label{eq:20}
  \end{eqnarray}
  By \refeq{eq:777}, $AInit$ returns an interval predicate
  $\underline{\Always \neg grd}$, which states that $\neg grd$ holds
  throughout the given interval, and the interval is
  non-empty. Condition \refeq{eq:88} is similar. Condition
  \refeq{eq:55} models the concurrent behaviour of processes $ap$ and
  $aq$. Process $ap$ either behaves as $\grd{grd} \ch \NE{\Always(m =
    1)}$ ($grd$ evaluates to true, then the behaviour of $m \asgn 1$
  holds) or $\grd{\neg grd} \ch \NE{\Always(m = 2)}$ ($\neg grd$
  evaluates to true, then the behaviour of $m \asgn 2$ holds, i.e.,
  the interval under consideration is non-empty and $m = 2$ holds
  throughout the interval). Process $aq$ is similar, but also models
  the assignments to
  $grd$. 

  Note that the points at which the intervals are chopped within
  \refeq{eq:55} and \refeq{eq:66} are unsynchronised. For example,
  suppose process $ap$ behaves as $\grd{grd} \ch \NE{\Always(m = 1)}$
  and $aq$ behaves as $\grd{b} \ch \underline{\Always grd}$ within
  interval $\Delta$ of stream $y$, i.e,. $(\grd{grd} \ch \NE{\Always(m
    = 1)}\land \grd{b} \ch \underline{\Always grd}).\Delta.y$ holds
  for some interval $\Delta$ and abstract stream $y$. By pointwise
  lifting, this is equivalent to $(\grd{grd} \ch \NE{\Always(m =
    1)}).\Delta.y \land (\grd{b} \ch \underline{\Always
    grd}).\Delta.y$. The two processes may now choose to split
  $\Delta$ independently. This includes the possibility of $\Delta$
  being split at the same point, which occurs if both guard
  evaluations are completed at the same time.
\end{example}

\section{A general theory of refinement}
\label{sec:gener-theory-refin}

We aim to verify data refinement between systems whose behaviours are
formalised by interval predicates. Hence, we present interval-based
data refinement (\refsec{sec:data-refinement-1}) and define
interval-based refinement relations (\refsec{sec:interval-relations}),
enabling formalisation of refinement relations in an interval-based
setting.  \refsec{sec:gener-forw-simul} presents our generalised proof
method, which is inspired by state-based forward simulation
techniques. \refsec{sec:decomp-simul} presents a number of
decomposition techniques for forward simulation.

\subsection{Data refinement}
\label{sec:data-refinement-1}
Existing frameworks for data refinement model concurrency as an
interleaving of the atomic system operations
\cite{BvW94,DB99,DB03,deRoever98}. This allows one to define a
system's execution using its set of operations. The traces of a system
after initialisation are generated by repeatedly picking an enabled
operation from the set non-deterministically then executing the
operation. Such execution models turn out to be inadequate for
reasoning about truly concurrent behaviour, e.g., about
\emph{transient} properties in the context of real-time systems
\cite{DH12MPC}. The methodology in this paper aims to allow modelling
of truly concurrent system behaviour. Each operation is associated
with exactly one of the system processes and execution of a system
(after initialisation) over an interval $\Delta$ is modelled by the
conjunction of the behaviours of each operation over $\Delta$ (see
\refex{ex:uv-ex-ip}).  It is possible to obtain interleaved
concurrency from our truly concurrent framework via the inclusion of
permissions \cite{Boy03,DDH12}.

Action refinement for true concurrency in a causal setting is studied
in \cite{M-CW01}, and a modal logic for reasoning about true
concurrency is given in \cite{BC10}. Frameworks for concurrent
refinement in real-time contexts have also been proposed (e.g.,
\cite{FUKH96,SH93}). We are however not aware of a method that allows
data refinement under true concurrency.

We let $Proc$ denote the set of all process identifiers. For $P
\subseteq Proc$ and $N, Z \subseteq Var$, respectively denoting the
sets of observable and representation variables, a \emph{system} is
defined by a tuple:
\begin{eqnarray*}
  C & \sdef & (CI, (COp_p)_{p:P}, CF)_{N,Z}
\end{eqnarray*}
where $CI : State_N \fun IntvPred_{Z}$ models the initialisation,
$COp_p \in IntvPred_{Z}$ for each $p \in P$ model the system
processes, and $CF : State_Z \fun State_N \fun \bool$ denotes system
finalisation.  The set of observable states at the start and end of an
execution of system $C$ is given by:
\begin{eqnarray*}
  obs_N.C & \sdef & 
  \left \{
      (\rho, \rho') : State_N \times State_N
       \begin{array}[c]{@{~~}|@{~~}l@{}}
         \exists \Delta : Intv, z : Stream_{Z} \st \\
         (\prev CI.\rho
         \land \bigwedge_{p:P} COp_p).\Delta.z \land \exists t : \Delta
         \st CF.(z.t).\rho'
  \end{array}\right\}
\end{eqnarray*}
\begin{definition}
  \label{def:dref}
  For $P \subseteq Proc$, an abstract system $A \sdef (AI,
  (AOP_p)_{p:P}, AF)_{N,Z}$ is \emph{data refined} by a concrete
  system $C\sdef (CI, (COP_p)_{p:P}, CF) _{N,Z}$, denoted $A \sref C$
  iff $obs_N.C \subseteq obs_N.A$.
\end{definition}

\noindent It is trivial to prove that $\sref$ is a preorder (i.e., a
reflexive, transitive relation).

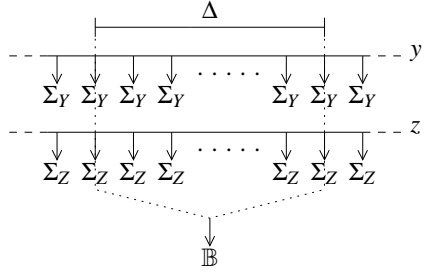
\begin{figure}[t]
  \centering
  \begin{minipage}[b]{0.49\columnwidth}
    \centering
    \scalebox{0.8}{\input{intv-rel.pspdftex}}
    \caption{Interval relation visualisation}
    \label{fig:inv-rel-vis}
  \end{minipage}
\end{figure}

Verification of \refdef{def:dref} directly is infeasible. In
state-based formalisms, data refinement is proved using
\emph{simulation}, which allows executions of the concrete system to
be matched to executions of the abstract \cite{deRoever98} (see
\reffig{fig:dref}). Previous work \cite{DDH12,DH12MPC} defines
operation refinement over a single state space. This cannot be used
for example to prove refinement between the programs in Figures
\ref{fig:abs-ref-atom} and \ref{fig:conc-ref-atom}. In this paper, we
develop simulation-based techniques for our interval-based framework
in \refsec{sec:gener-forw-simul}. The theory is based on interval
relations (\refsec{sec:interval-relations}), which enable one to
relate streams over two potentially different state spaces.

\subsection{Interval relations}
\label{sec:interval-relations}

Interval predicates enable one to reason about properties that take
time, however, only define properties over a single state
space. Proving data refinement via simulation requires one to relate
behaviours over a concrete state space to behaviours over an abstract
space. Hence, we combine the ideas behind state relations and interval
predicates and obtain \emph{interval relations}, which are relations
over an interval and two streams over potentially different state
spaces. The concept of interval relations is novel to this paper.

An \emph{interval relation} over $Y$ and $Z$ relates streams of $Y$
and $Z$ over intervals and is a mapping of type $IntvRel_{Y,Z}
\sdef Intv \fun Stream_{Y} \fun Stream_{Z} \fun \bool$. 
 \reffig{fig:inv-rel-vis} depicts a visualisation of an
interval relation over $Y, Z \subseteq Var$ where $z \in Stream_{Z}$
and $y \in Stream_Y$. Like interval predicates, we assume pointwise
lifting of operators over state and interval relations in the normal
manner. We extend interval predicate operators to interval relations,
for example:\smallskip

\hfill$\begin{array}{rcl}
  (R_1 \ch R_2).\Delta.y.z & \sdef &
  \begin{array}[t]{@{}l@{}}
    \left(\begin{array}[c]{@{}l@{}}
        \exists \Delta_1, \Delta_2 : Intv  \st
        (\Delta = \Delta_1
        \cup \Delta_2) \land  (\Delta_1 \adjoins \Delta_2) 
        \land R_1.\Delta_1.y.z \land R_2.\Delta_2.y.z
      \end{array}\right) \lor \\
    (\Inf.\Delta \land R_1.\Delta.y.z)
  \end{array} 
\end{array}$\hfill \smallskip {}

\noindent
A \emph{state relation} over $Y, Z \subseteq
Var$ 
is defined by its
characteristic function $StateRel_{Y,Z} \sdef State_Y \fun State_Z
\fun \bool$. Operators on state predicates may be extended to state
relations, e.g., for $r \in StateRel_{Y,Z}$ we define \smallskip

\hfill$\begin{array}{rcl}
  \Always r.\Delta.y.z  &  \sdef  &  \all t : \Delta \st r.(y.t).(z.t) 
\end{array}$
\hfill \smallskip 

\noindent If $R_1 \in IntvRel_{X, Y}$ and $R_2 \in IntvRel_{Y, Z}$
then for $\Delta \in Intv$, $x \in Stream_{X}$, $y \in
Stream_{Y}$, we define the composition of $R_1$ and $R_2$ as 
\smallskip 

\hfill$\begin{array}{rcl}
(R_1 \circ R_2).\Delta.x.z & \sdef &
\exists y : Stream_{Y} \st R_1.\Delta.x.y \land
R_2.\Delta.y.z
\end{array}$\hfill \smallskip 


\subsection{Generalised forward simulation}
\label{sec:gener-forw-simul}

In this section, we work towards an interval-based notion of forward
simulation, which is then shown to be a sufficient condition for
proving data refinement (\refdef{def:dref}).

\begin{figure}[t]
    \centering
    \input{intvfs.pspdftex}
    \caption{Visualisation of $\dataref{Y}{Z}{ref}{g}{h}$}
    \label{fig:intv-fs}
\end{figure}

We define simulation between abstract and concrete systems with
respect to an interval relation over the sets of representation
variables of the two systems. This definition requires that we define
equivalence between two streams over an interval. For streams $y$ and
$z$ and interval $\Delta$, we define a function\smallskip

\hfill$\begin{array}{rcl}
  \matches{\Delta}{y}{z}
  & \sdef &  (\Delta \dres y = \Delta \dres z)
\end{array}$\hfill\smallskip

\noindent where `$\dres$' denotes domain restriction. Thus
$\matches{\Delta}{y}z$ holds iff the states of $y$ and $z$
corresponding to $\Delta$ match, i.e., $\all t : \Delta \st y.t =
z.t$.  For $Y, Z \subseteq Var$, assuming that $g \in IntvPred_{Y}$
and $h \in IntvPred_{Z}$ model the abstract and concrete systems,
respectively, and that $ref \in IntvRel_{Y, Z}$ denotes the refinement
relation, we define a function $\dataref{Y}{Z}{ref}{g}{h}$ (see
\reffig{fig:intv-fs}), which denotes that $h$ \emph{simulates} $g$
with respect to $ref$.
\begin{eqnarray*}
  \dataref{Y}{Z}{ref}{g}{h} & \sdef & 
  \begin{array}[c]{@{}l@{}}
    \all z : Stream_{ Z}, \Delta, \Delta_0 :
    Intv, y_0 : Stream_{ Y} \st \\
    \qquad \qquad \begin{array}[t]{@{}l@{}}
      (\Delta_0 \adjoins \Delta) 
      \land
      ref.\Delta_0.y_0.z
      \land h.\Delta.z \imp \quad {} \\
      \hfill \exists y :  Stream_{ Y}
      \st \bmatches{\Delta_0}{y_0}{y} \land ref.\Delta.y.z \land g.\Delta.y
    \end{array}
  \end{array}  
\end{eqnarray*} 
Thus, if $\dataref{Y}{Z}{ref}{g}{h}$ holds, then for every concrete
stream $z$, interval $\Delta$ and abstract state $y$, provided that
\begin{enumerate}
\item $\Delta_0$ is an interval that immediately precedes $\Delta$,
\item $ref$ holds in the interval $\Delta_0$ between $y_0$
  and $z$, and
\item the concrete system (modelled by $h$) executes within $\Delta$
  in stream $z$
\end{enumerate}
then there exists an abstract stream $y$ that matches $y_0$ over
$\Delta_0$ such that
\begin{enumerate}
\item the abstract system executes over $\Delta$ in $y$, and
\item $ref$ holds between $y$ and $z$ over $\Delta$.
\end{enumerate}
A visualisation of $\dataref{Y}{Z}{ref}{g}{h}$ is given in
\reffig{fig:intv-fs} and is akin to matching a single non-stuttering
concrete step to an abstract step in state-based forward simulation
\cite{deRoever98}.  The following lemma establishes reflexivity and
transitivity properties for $\dataref{Y}{Z}{ref}{g}{h}$.
\begin{lemma} 
  \label{thm:dataref}
  Provided that $id.\sigma.\tau \sdef \sigma = \tau$.
  \begin{align}
    \dataref{X}{X}{\Always id}{g}{g} \TL{Reflexivity}
    \\
    \dataref{X}{Y}{ref_1}{f}{g} \land  \dataref{Y}{Z}{ref_2}{g}{h} 
    \ \ \imp \ \ & \dataref{X}{Z}{(ref_1 \circ
      ref_2)}{f}{h} \TL{Transitivity}
  \end{align}
\end{lemma}

Simulation is used to define an interval-based notion of \emph{forward
  simulation} as follows. 
\begin{definition}[Forward simulation]
  \label{def:forward-sim}
  Suppose $P \subseteq Proc$, $A \sdef (AI, (AOp_p)_{p:P}, AF)_{N,Y}$
  is an abstract system, $C \sdef (CI, (COp_p)_{p:P}, CF)_{N,Z}$ is a
  concrete system, and $ref \in IntvRel_{Y,Z}$. We say $ref$ is a
  \emph{forward simulation} from $A$ to $C$ iff
  $\dataref{Y}{Z}{ref}{\bigwedge_{p:P} AOp_p}{\bigwedge_{p:P} COp_p}$
  and both of the following hold:
  \begin{eqnarray}
    \label{eq:9}
    \all z : Stream_Z, \Delta : Intv, \sigma \in State_N \st
    CI.\sigma.\Delta.z  & \imp & 
    \exists y : Stream_Y \st AI.\sigma.\Delta.y \land ref.\Delta.y.z
    \\
    \begin{array}[b]{@{}l@{}}
      \all z : Stream_{ Z}, y : Stream_Y, \Delta :
      Intv , \sigma :  State_N \st \all t : \Delta \st \qquad \\
      \hfill 
      ref.\Delta.y.z
      \land CF.(z.t
      ).\sigma
    \end{array}
    & \imp & 
    AF.(y.t
    ).\sigma  
    \label{eq:4}
  \end{eqnarray}
\end{definition}

The following theorem establishes soundness of our forward simulation
rule with respect to interval-based data refinement.
\begin{theorem}[Soundness] 
  If $P \subseteq Proc$, $A \sdef (AI, (AOp_p)_{p:P}, AF)_{N,Y}$, and
  $C \sdef (CI, (COp_p)_{p:P}, CF)_{N,Z}$, then $A \sref C$ provided
  there exists a $ref \in IntvRel_{Y,Z}$ such that $ref$ is a forward
  simulation from $A$ to $C$.
\end{theorem}
\begin{proof}[Proof] 
  Suppose $\sigma, \sigma' \in State_N$, $z \in Stream_Z$ and $C$ has
  an execution depicted below, where $CI$ executes in interval
  $\Delta_0$ and $\bigwedge_{p:P} COp_p$ executes in $\Delta$. Note
  that $\bigwedge_{p:P} COp_p$ may or may not terminate, and hence,
  $\Delta$ may be infinite. To prove $A \sref C$, it suffices to prove
  that there exists a matching execution of $A$ starting in $\sigma$
  and ending in $\sigma'$.
  \begin{center}
    \scalebox{0.9}{\input{soundnessa.pspdftex}}
  \end{center}
  By \refeq{eq:9}, there exists a $y_0 \in Stream_Y$ such that
  $AI.\sigma.\Delta_0.y_0$ and $ref.\Delta_0.y_0.z$ hold recalling
  that $\Delta_0$ is the initial interval of execution. This is
  depicted in (A) below.  Now, because the simulation
  $\dataref{Y}{Z}{ref}{\bigwedge_{p:P} AOp_p}{\bigwedge_{p:P} COp_p}$
  holds, there exists a $y$ that matches $y_0$ over $\Delta_0$ such
  that both $(\bigwedge_{p:P} AOp_p).\Delta.y$ and $ref.\Delta.y.z$
  hold, as depicted in (B) below.
  \begin{center}
    \begin{tabular}[t]{c@{\qquad}c}
      (A) & (B) \smallskip \\
      \scalebox{0.9}{\input{soundnessb.pspdftex}}
      & 
      \scalebox{0.9}{\input{soundnessc.pspdftex}}
    \end{tabular}
  \end{center}
  Then, due to the finalisation assumption \refeq{eq:4}, there exists
  a finalisation of $A$ that results in $\sigma'$ as shown in (D)
  below.
  \begin{center}
    \begin{tabular}[t]{c@{\qquad}c}
      (D) \smallskip \\
      \raisebox{1.1em}{\scalebox{0.9}{\input{soundnesse.pspdftex}}}      
    \end{tabular}
  \end{center}
\end{proof}

\subsection{Decomposing simulations}
\label{sec:decomp-simul}

A benefit of state-based forward simulation \cite{deRoever98} is the
ability to decompose proofs and focus on individual steps of the
concrete system. Proof obligation $\dataref{Y}{Z}{ref}{g}{h}$ in the
interval-based forward simulation definition
(\refdef{def:forward-sim}) takes the entire interval of execution of
the concrete and abstract systems into account. Hence, we develop a
number of methods for simplifying proofs of
$\dataref{Y}{Z}{ref}{g}{h}$\ .  Decomposing
$\dataref{Y}{Z}{ref}{g}{h}$ directly is difficult due to the
existential quantification in the consequent. However, a formula of
the form $p \imp (\exists x \st q \land r)$ holds if both $p \imp
\exists x \st q$ and $\all x \st p \land q \imp r$ hold. Hence, we
obtain the following lemma.
\begin{lemma}
  \label{lem:decomp-fs}
  For any $Y, Z \subseteq Var$ and $ref \in IntvRel_{Y,Z}$,
  $\dataref{Y}{Z}{ref}{g}{h}$ holds if both of the following hold: 
  \begin{eqnarray}
    \begin{array}[b]{@{}l@{}}
      \all z : Stream_{ Z}, \Delta, \Delta_0 :
      Intv, y_0 : Stream_{ Y} \st \\
      \qquad \Delta_0 \adjoins \Delta 
      \land ref.\Delta_0.y_0.z
      \land h.\Delta.z
    \end{array}
    &  \imp &  \exists y :  Stream_{ Y}
    \st \bmatches{\Delta_0}{y_0}{y} \land ref.\Delta.y.z 
    \label{eq:ref1}\\
    \begin{array}[b]{@{}r@{}}
      \all z : Stream_{ Z}, \Delta :
      Intv, y : Stream_{ Y} \st  \\
      ref.\Delta.y.z
      \land h.\Delta.z
    \end{array} & \imp & g.\Delta.y
    \label{eq:ref2}
  \end{eqnarray}
\end{lemma}
By \refeq{eq:ref1}, if the refinement predicate $ref$ holds for an
abstract stream $y_0$ in an immediately preceding interval $\Delta_0$
and the concrete system executes in the current interval $\Delta$,
then there exists an abstract stream that matches $y_0$ over
$\Delta_0$ and $ref$ holds for $y$ over $\Delta$. By \refeq{eq:ref2}
for any abstract stream $y$, concrete stream $z$ and interval
$\Delta$, if the concrete system executes in $\Delta$ and forward
simulation holds between $y$ and $z$ for $\Delta$, then the behaviour
of the abstract system holds for $\Delta$ in $y$.

To simplify representation of intervals of the form in
\refeq{eq:ref1}, we introduce the following notation.
\begin{eqnarray*}
  h \Vdash_{Y, Z} ref & \sdef & \refeq{eq:ref1}
\end{eqnarray*}
The following lemma allows one to decompose proofs of the form given
in $h \Vdash_{Y, Z} ref$.
\begin{lemma}
  \label{thm:seq-comp}
  If $Y, Z \subseteq Var$, $g,g_1, g_2 \in IntvPred_{Z}$ and $ref
  \in IntvRel_{Y, Z}$, then each of the following holds.
  \begin{align}
    \TL{Sequential composition}
    g_1 \Vdash_{Y, Z} ref  \land g_2 \Vdash_{Y,Z} ref & \ \ \imp\ \  (g_1
    \ch g_2) \Vdash_{Y, Z} ref & \text{provided $ref$ joins}
    \\
    \TL{Iteration}
    g \Vdash_{Y, Z} ref & \ \ \imp\ \   g^\omega \Vdash_{Y, Z} ref & \text{provided $ref$ joins}
    \\
    (g_2 \Vdash_{Y,Z} ref) \land (g_1 \entails g_2) &\ \ \imp\ \ g_1 \Vdash_{Y, Z} ref 
    \TL{Weaken}
    \\
    (g \Vdash_{Y,Z} ref_1)  \lor (g \Vdash_{Y,Z} ref_2)  &\ \ \imp\ \ g
    \Vdash_{Y, Z} (ref_1 \lor ref_2)  
    \TL{Disjunction}
  \end{align}
\end{lemma}

Note that $ref$ can neither be weakened nor strengthened in the
trivial manner because it appears in both the antecendent and
consequent of the implication. If a refinement relation operates on two disjoint portions of the
stream, it is possible to split the refinement as follows:
\begin{lemma}[Disjointness]
  \label{thm:disj}
  Suppose $p \in Proc$, $W, X, Y, Z \subseteq Var$ such that $Y \cap Z
  = \emptyset$, $W \cup X = Y$ and $W \cap X = \emptyset$. If $g_1,
  g_2 \in IntvPred_{Z}$, $ref_W \in IntvRel_{W, Z}$, $ref_X \in
  IntvRel_{X, Z}$, and $\star \in \{\land, \lor\}$, then
  \begin{align}
    (g_1 \Vdash_{W,Z} ref_W) \land (g_2 \Vdash_{X,Z} ref_X) &\ \ \imp\
    \ (g_1 \land g_2) \Vdash_{Y, Z} (ref_W \star ref_X)
    \TL{Disjointness}
  \end{align}
\end{lemma}
\noindent Disjointness allows one to prove mixed refinement, where the
system states are split into disjoint subsets and different refinement
relations are used to verify refinement between these substates.

Proof obligation \refeq{eq:ref2} may also be simplified.  In
particular, for interval predicate $g$, interval $\Delta$ and streams
$y$ and $z$, we define $(g \project 1).\Delta.y.z \sdef g.\Delta.y $
and $(g \project 2).\Delta.y.z \sdef g.\Delta.z$, which allows one to
shorten \refeq{eq:ref2} to
\begin{eqnarray}
  \label{eq:1}
  ref \land (h \project 2) & \entails & (g \project 1)
\end{eqnarray}
Hence, proofs of refinement are reduced to proofs of implication
between the concrete and abstract state spaces. There are numerous
rules for decomposing proofs of the form in \refeq{eq:1} that exploit
rely/guarantee-style reasoning \cite{DH12iFM,DDH12}. 

\section{Fine-grained atomicity}
\label{sec:fine-grain-atom}
Interval-based reasoning provides the opportunity to incorporate
methods for non-deterministi\-cal\-ly evaluating expressions
\cite{BH10,HBDJ13}, which captures the possible low-level
interleavings (e.g., \reffig{fig:split}) at a higher-level of
abstraction. Methods for non-deterministically evaluating expressions
are given in \refsec{sec:appar-stat-eval}, and also appear in
\cite{BH10,HBDJ13,DH12,DH12MPC,DH12iFM}. Verification of data
refinement of our running example that combines non-deterministic
evaluation from \refsec{sec:appar-stat-eval} and the data refinement
rules from \refsec{sec:gener-theory-refin} is given in
\refsec{sec:non-atomic-guard}.

\subsection{Non-deterministically evaluating expressions}
\label{sec:appar-stat-eval}

Most hardware can only guarantee that at most one global variable can
be read in a single atomic step. Thus, in the presence of possibly
interfering processes and fine-grained atomicity, a model that assumes
expressions containing multiple shared variables can be evaluated in a
single state may not be implementable without the introduction of
contention inducing locks \cite{AL95,BvW99,JP08}. As we have done in
\reffig{fig:split}, one may split expression evaluation into a number
of atomic steps to make the underlying atomicity explicit. However,
this approach is undesirable as it causes the complexity of expression
evaluation to increase exponentially with the number of variables in
an expression --- evaluation of an expression with $n$ (global)
variables would require one to check $n!$ permutations of the read
order.

Interval-based reasoning enables one to incorporate methods for
non-deterministically evaluating state predicates over an evaluation
interval \cite{HBDJ13}, which allow the possible permutations in the
read order of variables to be considered at a high level of
abstraction. For this paper, we use \emph{apparent states evaluators},
which allow one to evaluate an expression $e$ with respect to the set
of states that are apparent to a process. Each variable of $e$ is
assumed to be read at most once, but at potentially different
instants, and hence, instead of evaluating $e$ in a single atomic
step, apparent states evaluations assume expression evaluation takes
time and considers the set of states that occur over the interval of
evaluation. An apparent state is generated by picking a value for each
variable from the set of actual values of the variable over the
interval of evaluation. For $\Delta \in Intv$ and $s \in Stream_{V}$,
we define:\smallskip

\hfill$\begin{array}{rcl}
  apparent.\Delta.s & \sdef & 
  \{\sigma : State_V \mid 
  \forall v : V \st \exists t : \Delta \st \sigma.v = s.t.v \} 
\end{array}$\hfill{}

\begin{example}
  \label{ex:apparent}
  Consider the statements $u \asgn 1 \ch v \asgn 1$ which we assume
  are executed over an interval $\Delta$ from an initial state that
  satisfies $u, v = 0, 0$. The set of states that actually occur over
  this interval is hence\smallskip

  \hfill $\{\{ u \mapsto 0, v \mapsto 0 \}, \{
  u \mapsto 1, v \mapsto 0 \}, \{ u \mapsto 1, v \mapsto 1 \} \}$
  \hfill \smallskip 

  \noindent Evaluation of $u < v$ in the set of actual states above
  always results in $false$.  Assuming no other (parallel)
  modifications to $u$ and $v$, for some stream $s$ over $\{u,v\}$,
  the set of apparent states corresponding to $\Delta$ is: \smallskip
  
  \hfill$\begin{array}{rcl}
    apparent.\Delta.s
    & = &   \left\{\begin{array}{@{}l@{}}
        \{ u \mapsto 0, v \mapsto 0 \},
        \{ u \mapsto 1, v \mapsto 1 \}, 
        \{ u \mapsto 0, v \mapsto 1 \},
        \{ u \mapsto 1, v \mapsto 0 \} 
      \end{array}\right\}
  \end{array}$ \hfill \smallskip
  
  \noindent where the additional apparent state $ \{ u \mapsto 0, v
  \mapsto 1 \}$ may be obtained by reading $u$ with value $0$ (in the
  initial state) and $v$ with value $1$ (after both
  modifications). Unlike the actual states evaluation, $u < v$ may
  result in $false$ when evaluating in the apparent states. Note that
  ${v = v}$ still only has one possible value, $true$, i.e., apparent
  states evaluation assumes that the same value of $v$ is used for
  both occurrences of $v$.
\end{example}

Two useful operators for a sets of apparent states evaluation allow
one to formalise that $c$ \emph{definitely} holds (denoted $\Def c$)
and $c$ \emph{possibly} holds (denoted $\Pos c$), which are defined as
follows.
\begin{align*}
  (\Def c).\Delta.s \ \   \sdef \ \  &     \all \sigma :
  apparent.\Delta.s \st  c.\sigma & 
  (\Pos c).\Delta.s \ \   \sdef \ \  &    \exists \sigma :
  apparent.\Delta.s \st  c.\sigma
\end{align*}
The following lemma states a relationship between \emph{definitely}
and \emph{always} properties, as well as between \emph{possibly} and
\emph{sometime} properties \cite{HBDJ13}.  Note that both $\Def c
\entails \Always c$ and $\Eventually c \entails \Pos c$ hold, but the
converse of both properties are not necessarily true.

\begin{example} 
  We now instantiate the guard evaluations of the form $\grd{c}$
  within \refeq{eq:55} and \refeq{eq:66}. In particular, a guard $c$
  holds if it is possible to evaluate the variables of $c$ (at
  potentially different instants) so that $c$ evaluates to
  $true$. Therefore, the semantics of the evaluation of a guard $c$ is
  formalised by $\Pos c$ and we obtain the following interval
  predicates for \refeq{eq:55} and \refeq{eq:66}. 
  \begin{eqnarray}
      (\Pos grd \ch \NE{\Always(m =  1)} \lor \Pos \neg grd \ch
      \NE{\Always(m =  2)})
      & \land &
      (\Pos b \ch \underline{\Always grd} \lor \Pos \neg b 
      )
      \label{eq:5}
      \\
      \left(
        \begin{array}[c]{@{}l@{}}
          \Pos (u < v) \ch \NE{\Always(m =  1)} \lor \\
          \Pos (u \geq v) \ch
            \NE{\Always(m =  2)}
          \end{array}
        \right)
        & \land &
        \left(
          \begin{array}[c]{@{}l@{}}
            \Pos (0 < u) \ch \underline{\Always (v = \infty)} \lor \\
            \Pos (0 \geq u) \ch
            \underline{\Always (v = -\infty)}
          \end{array}
        \right)
    \label{eq:6}
  \end{eqnarray}
  Note that interval predicate $\Pos (u < v)$ is equivalent
  to\smallskip

  \hfill$\exists k_u, k_v
  \st ((\Eventually (k_u = u) \ch \Eventually(k_v = v)) \lor
  (\Eventually(k_v = v) \ch \Eventually(k_u = u))) \ch (k_u < k_v)$
  \hfill \smallskip

  \noindent Hence, the formalisation in \refeq{eq:6} accurately
  captures the fine-grained behaviour of \reffig{fig:conc-ref-atom}
  without having to explicitly decompose the guard evaluation at
  $cp_1$ into individual reads as done in \reffig{fig:split}.
\end{example}

The theory in \cite{HBDJ13} allows one to relate different forms of
non-deterministic evaluation. For example, both $\Def c \entails
\Always c$ and $\Eventually c \entails \Pos c$ hold. To strengthen the
implication to an equivalence, one must introduce additional
assumptions about the stability of the variables of $c$. Because
adjoining intervals are disjoint, the definition of stability must
refer to the value of $c$ at the end of an immediately preceding
interval \cite{DH12MPC,DDH12,DH12iFM}. For a state predicate $c$,
interval $\Delta$ and stream $s$, we define \smallskip

\hfill$\begin{array}[t]{rcl}
prev.c.\Delta.s & \sdef & \exists
\Delta' : Intv \st \Delta' \adjoins \Delta \land \NE{\Always c}.\Delta'.s\label{eq:11}
\end{array}$\hfill\smallskip

Variable $v$ is stable over a $\Delta$ in $s$ (denoted
$\mathsf{stable}.v.\Delta.s$) iff the value of $v$ does not change
from its value over some interval that immediately precedes $\Delta$. A
set of variables $V$ is \emph{stable} in $\Delta$ (denoted
$\mathsf{stable}.V.\Delta$) iff each variable in $V$ is stable in
$\Delta$. Thus, we define: 
$$
\begin{array}[t]{rclrcl}
  \mathsf{stable}.v.\Delta.s & \sdef & 
  \exists k : Val \st (prev.(v = k) \land \Always (v = k)).\Delta.s
  & \qquad
  \qquad \mathsf{stable}.V.\Delta & \sdef & \all v : V \st \mathsf{stable}.v.\Delta
\end{array}
$$
Note that every variable is stable in an empty interval and the empty
set of variables is stable in any interval, i.e., both
$\mathsf{stable}.V.\emptyset$ and $\mathsf{stable}.\emptyset.\Delta$ hold trivially.

We let $vars.c$ denote the free variables of state predicate $c$.  The
following lemma states that if all but one variable of $c$ is stable
over an interval $\Delta$, then $c$ definitely holds in $\Delta$ iff
$c$ always holds in $\Delta$, and that $c$ possibly holds in $\Delta$
iff $c$ holds sometime in $\Delta$ \cite{HBDJ13}.
\begin{lemma}
  \label{lem:st-grd-var}
  For a state predicate $c$ and variable $v$, 
  $\mathsf{stable}.(vars.c
  \bs \{v\}) \entails (\Def c = \Always c) \land (\Pos c = \Eventually
  c)$.
\end{lemma}

\begin{example}
  For our running example, by \reflem{lem:st-grd-var}, it is possible
  to simplify \refeq{eq:5} and \refeq{eq:6} and replace each
  occurrence of `$\Pos$' by `$\Eventually$' as follows:
  \begin{align}
    (\Eventually grd \ch  \NE{\Always(m =  1)}\lor \Eventually \neg grd \ch
    \NE{\Always(m =  2)}) \ \ \land\ \ & (\Eventually b \ch \underline{\Always grd}
    \lor \Eventually \neg b 
    )
    \TL{Abs-IP}
    \\
    \left(
    \begin{array}[c]{@{}l@{}}
      \Eventually (u < v) \ch \NE{\Always(m =  1)} \lor \\
      \Eventually (u \geq v)
      \ch \NE{\Always(m =  2)}
    \end{array}
  \right) \ \ \land\ \ & \left(
    \begin{array}[c]{@{}l@{}}
      \Eventually (0 < u) \ch
      \underline{\Always (v = \infty)} \lor \\
      \Eventually (0 \geq u) \ch
      \underline{\Always (v = -\infty)}
    \end{array}
    \right) \TL{Conc-IP}
  \end{align}
\end{example}

\subsection{Data refinement example}
\label{sec:non-atomic-guard}

We assume the representation variables of the abstract and concrete
programs are given by $Y \subseteq Var$ and $Z \subseteq Var$,
respectively and prove forward simulation using $\NE{\Always uv}$
(recalling that relation $uv$ is defined in \refsec{sec:background}),
which requires that we prove 
\begin{eqnarray}
  \dataref{Y}{Z}{\NE{\Always uv}}{\trefeq{Abs-IP}}{\trefeq{Conc-IP}} 
  \label{eq:13}
\end{eqnarray}
and both of the following: 
\begin{eqnarray}
  \begin{array}[b]{@{}r@{}}
    \all \Delta : Intv, z : Stream_Z, \sigma : State_N \st  \\
    CInit.\sigma.\Delta.z
  \end{array}
  & \imp &
  \exists y : Stream_Y \st AInit.\sigma.\Delta.y \land \NE{\Always uv}.\Delta.y.z
  \label{eq:12}
  \\
  \label{eq:21}
  \begin{array}[b]{@{}l@{}}
    \all z : Stream_{ Z}, y : Stream_Y, \Delta :
    Intv , \sigma :  State_N \st \all t : \Delta \st  \\
    \hfill 
    \NE{\Always uv}.\Delta.y.z
    \land CFin.(z.t).\sigma
  \end{array}
  & \imp & 
  AFin.(y.t).\sigma  
\end{eqnarray}
The proofs of \refeq{eq:12} and \refeq{eq:21} are trivial. To prove
\refeq{eq:13}, we use \reflem{lem:decomp-fs}, which requires that we
show that both of the following hold. Recall that $uv$ is the state
relation defined in \refsec{sec:interv-based-reas}.
\begin{eqnarray}
  \label{eq:16}
  & \trefeq{Conc-IP} \Vdash_{Y,Z} \NE{\Always uv} \\
  \label{eq:17}
  & \NE{\Always uv} \land \trefeq{Conc-IP}\project 2 \entails
  \trefeq{Abs-IP} \project 1
\end{eqnarray}
The proof of \refeq{eq:16} is trivial. Expanding the definitions of
\trefeq{Abs-IP} and \trefeq{Conc-IP}, then applying some straightforward
propositional logic, \refeq{eq:17}, holds if both of the following
hold.
\begin{eqnarray}
  \NE{\Always uv} \land  \left(
    \begin{array}[c]{@{}l@{}}
      \Eventually (0 < u) \ch \underline{\Always (v = \infty)} \lor\\
      \Eventually (0 \geq u) \ch \underline{\Always (v = -\infty)}
    \end{array}
  \right) 
  \project 2 
  & \entails & 
  (\Eventually b \ch \underline{\Always grd} \lor \Eventually \neg b 
  ) 
  \project 1 
  \label{eq:7}
  \\
  \NE{\Always uv} \land \left(
    \begin{array}[c]{@{}l@{}}
      \Eventually (u < v) \ch \NE{\Always(m =  1)} \lor \\
      \Eventually (u \geq v)
      \ch \NE{\Always(m =  2)}
    \end{array}
  \right)  
  \project 2 
  & \entails & 
  \left(
    \begin{array}[c]{@{}l@{}}
      \Eventually grd \ch \NE{\Always(m =  1)} \lor \\
      \Eventually \neg grd
      \ch \NE{\Always(m =  2)}
    \end{array}
  \right)\project 1 
  \label{eq:3}
\end{eqnarray}
Condition \refeq{eq:7} is proved in a straightforward manner as
follows and uses the fact that $\Always(u < \infty)$ holds throughout
the execution of \reffig{fig:conc-ref-atom}.
\begin{derivation}
  \step{ \NE{\Always uv} \land ( \Eventually (0 < u) \ch \underline{\Always
      (v = \infty)} \lor \Eventually (0 \geq u) \ch \underline{\Always
      (v = -\infty)} ) \project 2}

  \trans{\entails}{distribute projection, logic and $\Always(u < \infty)$}

  \step{\Always uv \land ((\Eventually (0 < u) \project 2\ch \underline{\Always
      (u < v)} \project 2) \lor (\Eventually (0 \geq u)\project 2 \ch \underline{\Always
      (u \geq v)}  \project 2))}

  \trans{\entails}{distribute $\land$, $\Always uv$ splits}

  \step{( \Always uv \land \Eventually (0 < u) \project 2)) \ch 
    (\Always uv \land \underline{\Always 
      (u < v)} \project 2) \lor }
  \step{( \Always uv \land \Eventually (0 \geq u) \project 2) \ch (
    \Always uv \land\underline{\Always (u \geq v)}  \project 2)}

  \trans{\entails}{use $\Always uv$}

  \step{(\Eventually b \project 1 \ch 
    \underline{\Always 
       grd} \project 1) \lor (\Eventually \neg b) \project 1
   }

  \trans{\equiv}{distribute projection}

  \step{((\Eventually b \ch 
    \underline{\Always 
       grd}) \lor \Eventually \neg b
     )  \project 1}

\end{derivation}
The proof of \refeq{eq:3} has a similar structure, and hence, its
details are elided. 

The example verification demonstrates many of the benefits of using
interval-based reasoning to prove data refinement between concurrent
systems.  The proofs themselves are succinct (and consequently more
understandable) because the reasoning is performed at a high level of
abstraction. Expression evaluation is assumed to take time and
evaluation operators such as `$\Eventually$' and `$\Pos$' are used to
capture the inherent non-determinism that results from concurrent
executions during the interval of evaluation. Furthermore, the
translation of the program in \reffig{fig:conc-ref-atom} to the
lower-level program \reffig{fig:split} that makes the non-determinism
for evaluating reads explicit is not necessary. Instead, one is able
to provide a semantics for the program in \reffig{fig:conc-ref-atom}
directly. Finally, unlike a state-based forward simulation proof,
which requires that a verifier explicitly decides which of the
concrete steps are non-stuttering, then find a corresponding abstract
step for each non-stuttering step, interval-based reasoning allows one
to remove this analysis step altogether.

\section{Conclusions}

Interval-based frameworks are effective for reasoning about
fine-grained atomicity and true concurrency in the presence of both
discrete and continuous properties.  The main contribution of this
paper is the development of generalised methods for proving data
refinement using interval-based reasoning.  A simulation rule for
proving data refinement is developed and soundness of the rule with
respect to the data refinement definition is proved. Our simulation
rule allows the use of refinement relations between streams over two
state spaces within an interval, generalising traditional refinement
relations, which only relate two states. Using interval-based
reasoning enables one to incorporate methods for non-deterministically
evaluating expressions, which in combination with our simulation rules
are used to verify data refinement of a simple concurrent
program. 

Over the years, numerous theories for data refinement have been
developed. As far as we are aware, two of these are based on
interval-based principles similar to ours. A framework that combines
interval temporal logic and refinement has been defined by B{\"a}umler
et al \cite{BSTR11}, but their execution model explicitly interleaves
a component and its environment. As a result, our high-level
expression evaluation operators cannot be easily incorporated into
their framework. Furthermore, refinement is defined in terms of
relations between the abstract and concrete states. Broy presents
refinement between streams of different types of timed systems (e.g.,
discrete vs. continuous systems) \cite{Bro01}; however, these methods
do not consider interval-based reasoning. An interesting direction of
future work would be to consider a model that combines our methods
with theories for refinement between different abstractions of time.

\smallskip

\noindent
\textbf{Acknowledgements} This work is sponsored by EPSRC Grant
EP/J003727/1. We thank our anonymous reviewers for their numerous
insightful comments. In particular, one reviewer who pointed out a
critical flaw in one of our lemmas.

\bibliographystyle{eptcs}
\bibliography{thesis,jreferences}

\appendix

\section{Proofs of lemmas}

\noindent {\bf Lemma} 
  \refeq{thm:dataref} 
  Provided that $id.\sigma.\tau \sdef \sigma = \tau$.
  \begin{align}
    \dataref{X}{X}{\Always id}{g}{g} \TL{Reflexivity}
    \\
    \dataref{X}{Y}{ref_1}{f}{g} \land  \dataref{Y}{Z}{ref_2}{g}{h} 
    \ \ \imp \ \ & \dataref{X}{Z}{(ref_1 \circ
    ref_2)}{f}{h} \TL{Transitivity}
  \end{align}
\begin{proof}
  The proof of \refeq{Reflexivity} is trivial. We prove
  \refeq{Transitivity} as follows, where we assume that $\Delta_0,
  \Delta \in Intv$ such that $\Delta_0 \adjoins \Delta$, $x_0 \in
  Stream_{X}$ and $z \in Stream_{Z}$ are arbitrarily
  chosen. We have:
  \begin{derivation}
    \step{(ref_1 \circ ref_2).\Delta_0.x_0.z \land h.\Delta.z}
    \trans{=}{definition of $\circ$ and logic}
    \step{\exists y_0: Stream_{ Y} \st ref_1.\Delta_0.x_0.y_0 \land ref_2.\Delta_0.y_0.z \land h.\Delta.z}
  \end{derivation}
  Hence, for an arbitrarily chosen $y_0 \in Stream_{ Y}$, we prove the
  following.
  \begin{derivation}
    \step{ref_1.\Delta_0.x_0.y_0 \land ref_2.\Delta_0.y_0.z \land h.\Delta.z}

    \trans{\imp}{assumption $\dataref{Y}{Z}{ref_2}{g}{h}$}

    \step{ref_1.\Delta_0.x_0.y_0 
      \land \exists y : Stream_{ Y} \st \bmatches{\Delta_0}{y_0}{y} \land
      ref_2.\Delta.y.z \land g.\Delta.y} 

    \trans{=}{logic assuming freeness of $y$}

    \step{\exists y : Stream_{ Y} \st  
        \bmatches{\Delta_0}{y_0}{y} \land
      ref_2.\Delta.y.z \land ref_1.\Delta_0.x_0.y_0 \land g.\Delta.y} 
    
    \trans{=}{$Stream_{ Y}$ contains all possible streams}

    \step{\exists y_1 : Stream_{ Y} \st  
        \bmatches{\Delta_0}{y_0}{y_1} \land \bmatches{\Delta}{y}{y_1} \land
      ref_2.\Delta.y.z \land ref_1.\Delta_0.x_0.y_0 \land g.\Delta.y} 

    \trans{=}{use $\matches{\Delta_0}{y_0}{y_1}$ and $\matches{\Delta}{y}{y_1}$}

    \step{\exists y_1 : Stream_{ Y} \st  
      ref_2.\Delta.y_1.z \land ref_1.\Delta_0.x_0.y_1 \land g.\Delta.y_1}

    \trans{\imp}{logic, assumption $\dataref{X}{Y}{ref_1}{f}{g}$}

    \step{\exists y_1 : Stream_{ Y}, x : Stream_{ X} \st
        \bmatches{\Delta_0}{x_0}{x} \land
      ref_2.\Delta.y_1.z \land ref_1.\Delta.x.y_1 \land f.\Delta.x} 

    \trans{\imp}{definition of $\circ$}

    \step{\exists x : Stream_{ X} \st \bmatches{\Delta_0}{x_0}{x} \land
      (ref_1 \circ ref_2).\Delta.x.z \land f.\Delta.x \hfill\qedhere}

  \end{derivation}

\end{proof}

\noindent
{\bf Lemma \refeq{thm:seq-comp}}
  Suppose $Y, Z \subseteq Var$ such that $Y \cap Z = \emptyset$,
  $g,g_1, g_2 \in IntvPred_{Z}$ and $ref \in IntvRel_{Y, Z}$. Then:
  \begin{align}
    \TL{Sequential composition}
    g_1 \Vdash_{Y, Z} ref  \land g_2 \Vdash_{Y,Z} ref & \ \ \imp\ \  (g_1
    \ch g_2) \Vdash_{Y, Z} ref &     \text{provided $ref$ joins}
    \\
    \TL{Iteration}
    g \Vdash_{Y, Z} ref & \ \ \imp\ \   g^\omega \Vdash_{Y, Z} ref &
    \text{provided $ref$ joins}
    \\
    (g_2 \Vdash_{Y,Z} ref) \land (g_1 \entails g_2) &\ \ \imp\ \ g_1 \Vdash_{Y, Z} ref 
    \TL{Weaken}
    \\
    (g \Vdash_{Y,Z} ref_1)  \lor (g \Vdash_{Y,Z} ref_2)  &\ \ \imp\ \ g
    \Vdash_{Y, Z} (ref_1 \lor ref_2)  
    \TL{Disjunction}
  \end{align}
\begin{proof}[Proof of \refeq{Sequential composition}]
  For an arbitrarily chosen $\Delta_0, \Delta \in Intv$ such that
  $\Delta_0 \adjoins \Delta$, $y_0 \in State_Y$ and $z \in Stream_{ Z}$,
  we have the following calculation.
  \begin{derivation}
    \step{
      ref.\Delta_0.y_0.z \land (g_1 \ch
      g_2).\Delta.z
    }

    \trans{=}{definition of `;', logic}

    \step{\exists \Delta_1, \Delta_2: Intv \st (\Delta_1 \cup
      \Delta_2 = \Delta) \land (\Delta_1 \adjoins \Delta_2) \land
      ref.\Delta_0.y_0.z \land g_1.\Delta_1.z \land
      g_2.\Delta_2.z }

    \trans{\imp}{$\Delta_0 \adjoins \Delta$ and $\Delta_1 \in
      prefix.\Delta$, therefore $\Delta_0 \adjoins \Delta_1$}

    \trans{}{assumption $g_1 \Vdash_{Y,Z} ref$} 

    \step{ \exists \Delta_1, \Delta_2: Intv \st
      \begin{array}[t]{@{}l@{}}
        (\Delta_1 \cup
        \Delta_2 = \Delta) \land (\Delta_1 \adjoins \Delta_2) \land 
        (\exists y_1 : Stream_{Y}\st \bmatches{\Delta_0}{y_0}{y_1} \land
        ref.\Delta_1.y_1.z) \land
        g_2.\Delta_2.z
      \end{array}
    }

    \trans{=}{logic}

    \step{ \exists \Delta_1, \Delta_2: Intv, y_1 : Stream_{Y} \st
      \begin{array}[t]{@{}l@{}}
        (\Delta_1 \cup \Delta_2 = \Delta) \land (\Delta_1 \adjoins
        \Delta_2) \land 
        \bmatches{\Delta_0}{y_0}{y_1}
        \land ref.\Delta_1.y_1.z \land g_2.\Delta_2.z
      \end{array}
    }

    \trans{=}{$\Delta_1 \adjoins \Delta_2$ and assumption $g_2
      \Vdash_{Y,Z} ref$}

    \step{\exists \Delta_1, \Delta_2: Intv, y_1, y_2 : Stream_{Y} \st
      \begin{array}[t]{@{}l@{}}
        (\Delta_1 \cup \Delta_2 = \Delta) \land (\Delta_1 \adjoins
        \Delta_2) \land \\
        \bmatches{\Delta_0}{y_0}{y_1}
        \land ref.\Delta_1.y_1.z \land  
        \bmatches{\Delta_1}{y_1}{y_2} \land ref.\Delta_2.y_2.z
      \end{array}
    }

    \trans{\imp}{pick $y_3$ such that $\matches{\Delta_0 \cup
      \Delta_1}{y_1}{y_3}$ and $\matches{\Delta_2}{y_2}{y_3}$}

    \step{\exists \Delta_1, \Delta_2: Intv, y_3 : Stream_{Y} \st
      \begin{array}[t]{@{}l@{}}
        (\Delta_1 \cup \Delta_2 = \Delta) \land (\Delta_1 \adjoins
        \Delta_2) \land 
        \bmatches{\Delta_0}{y_0}{y_3}
        \land ref.\Delta_1.y_3.z \land ref.\Delta_2.y_3.z
      \end{array}
    }

    \trans{=}{definition}

    \step{\exists y_3 : Stream_{Y} \st \bmatches{\Delta_0}{y_0}{y_3} \land (ref \ch
      ref).\Delta.y_3.z}

    \trans{\imp}{$ref$ joins}

    \step{\exists y_3 : Stream_{Y} \st \bmatches{\Delta_0}{y_0}{y_3} \land ref.\Delta.y_3.z \hfill\qedhere}

  \end{derivation}
\end{proof}
\begin{proof}[Proof of \refeq{Iteration}]
  This followings by unfolding ${}^\omega$ and has a similar structure
  to \refeq{Sequential composition}. \hfill \qedhere
\end{proof}
\begin{proof}[Proof of \refeq{Weaken}]
  For an arbitrarily chosen $\Delta_0, \Delta \in Intv$ such that
  $\Delta_0 \adjoins \Delta$, $y_0 \in State_Y$ and $z \in Stream_{Z}$,
  we have the following calculation.

  \begin{derivation*}
    \step{ref.\Delta_0.y_0.z \land g_1.\Delta.z }

    \trans{\imp}{assumption $g_1 \entails g_2$}

    \step{ref.\Delta_0.y_0.z \land g_2.\Delta.z }

    \trans{\imp}{assumption $g_2 \Vdash_{Y,Z} ref$}

    \step{\exists y : Stream_{Y} \st \bmatches{\Delta_0}{y_0}{y} \land
      ref.\Delta.y.z \hfill\qedhere}
        
  \end{derivation*}  
\end{proof}
\begin{proof}[Proof of \refeq{Disjunction}]

  \begin{derivation}
    \step{}
    \step{(ref_1 \lor ref_2).\Delta_0.y_0.z \land g.\Delta.z} 

    \trans{=}{logic}

    \step{(ref_1.\Delta_0.y_0.z \land g.\Delta.z) \lor
      (ref_2.\Delta_0.y_0.z \land g.\Delta.z)} 

    \trans{\imp}{assumption $(g \Vdash_{Y,Z} ref_1)  \lor (g
      \Vdash_{Y,Z} ref_2)$, logic}

    \step{\exists y_1, y_2 : Stream_{Y} \st (\bmatches{\Delta_0}{y_0}{y_1} \land
      ref_1.\Delta.y_1.z) \lor
      (\bmatches{\Delta_0}{y_0}{y_2} \land
      ref_2.\Delta.y_2.z)
    }

    \trans{\imp}{logic}

    \step{\exists y : Stream_{Y} \st \bmatches{\Delta_0}{y_0}{y} \land
      (ref_1 \lor ref_2).\Delta.y.z
    \hfill\qedhere}
  \end{derivation}
  
\end{proof}

For streams $s_1$ and $s_2$, we define $ s_1 \Cup s_2 \sdef \lambda t
: \Phi \st s_1.t \cup s_2.t$. If the state spaces corresponding to
$s_1$ and $s_2$ are disjoint, then for each $t \in \Phi$, $(s_1 \Cup
s_2).t$ is a state and hence $s_1 \Cup s_2$ is a stream.\smallskip

\noindent {\bf Lemma \refeq{thm:disj}}\textrm{(Disjointness)}
Suppose $p \in Proc$, $W, X, Y, Z \subseteq Var$ such that $Y \cap Z =
\emptyset$, $W \cup X = Y$ and $W \cap X = \emptyset$. Further suppose
that $g_1, g_2 \in IntvPred_{Z}$, $ref_W \in IntvRel_{W, Z}$, $ref_X
\in IntvRel_{X, Z}$, and $\star \in \{\land, \lor\}$. Then
  \begin{align}
    (g_1 \Vdash_{W,Z} ref_W) \land (g_2 \Vdash_{X,Z} ref_X) &\ \ \imp\
    \ (g_1 \land g_2) \Vdash_{Y, Z} (ref_W \star ref_X)
    \TL{Disjointness}
  \end{align}
\begin{proof}
  Because $W \cup X = Y$ and $W \cap X = Y$, for any $y_0 \in
  Stream_{Y}$, we have that $y_0 = w_0 \Cup x_0$ for some $w_0 \in
  Stream_W$, $x_0 \in Stream_{X}$. Then for any $z \in Stream_{Z}$,
  $\Delta_0, \Delta \in Intv$ such that $\Delta_0 \adjoins
  \Delta$, we have the following calculation:
  \begin{derivation}
    \step{(ref_W \star ref_X).\Delta_0.y_0.z \land
      (g_1 \land g_2).\Delta.z}

    \trans{\imp}{assumption $y_0 = w_0 \Cup x_0$}

    \step{(ref_W.\Delta_0.w_0.z \star ref_X.\Delta_0.x_0.z) \land
      (g_1 \land g_2).\Delta.z}

    \trans{\imp}{$\land$ distributes over $\star$, logic}

    \step{(ref_W.\Delta_0.w_0.z  \land g_1.\Delta.z) \star 
      (ref_X.\Delta_0.x_0.z \land g_2.\Delta.z)}
    
    \trans{\imp}{assumption $(g_1 \Vdash_{W,Z} ref_W) \land (g_2
      \Vdash_{X,Z} ref_X)$}

    \step{(\exists w : Stream_W\st \bmatches{\Delta_0}{w_0}{w} \land
      ref_W.\Delta.w.z) \star (\exists x : Stream_{X}\st
      \bmatches{\Delta_0}{x_0}{x} \land ref_X.\Delta.x.z)} 

    \trans{=}{logic, assumption $W \cap X = \emptyset$}

    \step{\exists w : Stream_W, x : Stream_{X} \st
      \bmatches{\Delta_0}{w_0 \Cup x_0}{w \Cup x} \land
      (ref_W.\Delta.w.z  \star ref_X.\Delta.x.z)} 

    \trans{=}{logic, assumption $y_0 = w_0 \Cup x_0$}

    \step{\exists w : Stream_W, x : Stream_{X} \st
      \bmatches{\Delta_0}{y_0}{w \Cup x} \land
      (ref_W \star ref_X).\Delta.(w \Cup x).z}

    \trans{=}{$W \cup X = Y$ and $W \cap X = \emptyset$}

    \step{\exists y : Stream_{Y} \st
      \bmatches{\Delta_0}{y_0}{y} \land
      (ref_W \star ref_X).\Delta.y.z \hfill \qedhere}

  \end{derivation}
\end{proof}



\end{document}

%% file: dref.pspdftex
\begin{picture}(0,0)%
\includegraphics{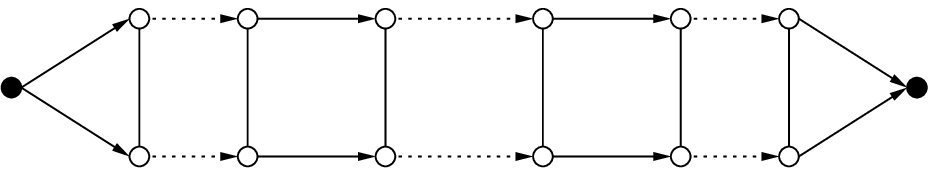}%
\end{picture}%
\setlength{\unitlength}{4144sp}%
\begingroup\makeatletter\ifx\SetFigFont\undefined%
\gdef\SetFigFont#1#2#3#4#5{%
  \reset@font\fontsize{#1}{#2pt}%
  \fontfamily{#3}\fontseries{#4}\fontshape{#5}%
  \selectfont}%
\fi\endgroup%
\begin{picture}(4246,1021)(2153,-1925)
\put(4906,-1051){\makebox(0,0)[b]{\smash{{\SetFigFont{10}{12.0}{\rmdefault}{\mddefault}{\updefault}{\color[rgb]{0,0,0}$id$}%
}}}}
\put(3421,-1456){\makebox(0,0)[b]{\smash{{\SetFigFont{10}{12.0}{\rmdefault}{\mddefault}{\updefault}{\color[rgb]{0,0,0}$ref$}%
}}}}
\put(3556,-1861){\makebox(0,0)[b]{\smash{{\SetFigFont{10}{12.0}{\rmdefault}{\mddefault}{\updefault}{\color[rgb]{0,0,0}$cp_i$}%
}}}}
\put(3556,-1051){\makebox(0,0)[b]{\smash{{\SetFigFont{10}{12.0}{\rmdefault}{\mddefault}{\updefault}{\color[rgb]{0,0,0}$ap_i$}%
}}}}
\put(4051,-1456){\makebox(0,0)[b]{\smash{{\SetFigFont{10}{12.0}{\rmdefault}{\mddefault}{\updefault}{\color[rgb]{0,0,0}$ref$}%
}}}}
\put(2926,-1456){\makebox(0,0)[b]{\smash{{\SetFigFont{10}{12.0}{\rmdefault}{\mddefault}{\updefault}{\color[rgb]{0,0,0}$ref$}%
}}}}
\put(4906,-1861){\makebox(0,0)[b]{\smash{{\SetFigFont{10}{12.0}{\rmdefault}{\mddefault}{\updefault}{\color[rgb]{0,0,0}$cp_j$}%
}}}}
\put(4501,-1456){\makebox(0,0)[b]{\smash{{\SetFigFont{10}{12.0}{\rmdefault}{\mddefault}{\updefault}{\color[rgb]{0,0,0}$ref$}%
}}}}
\put(5626,-1456){\makebox(0,0)[b]{\smash{{\SetFigFont{10}{12.0}{\rmdefault}{\mddefault}{\updefault}{\color[rgb]{0,0,0}$ref$}%
}}}}
\put(6211,-1726){\makebox(0,0)[b]{\smash{{\SetFigFont{10}{12.0}{\rmdefault}{\mddefault}{\updefault}{\color[rgb]{0,0,0}$CFin$}%
}}}}
\put(6211,-1231){\makebox(0,0)[b]{\smash{{\SetFigFont{10}{12.0}{\rmdefault}{\mddefault}{\updefault}{\color[rgb]{0,0,0}$AFin$}%
}}}}
\put(2341,-1726){\makebox(0,0)[b]{\smash{{\SetFigFont{10}{12.0}{\rmdefault}{\mddefault}{\updefault}{\color[rgb]{0,0,0}$CInit$}%
}}}}
\put(2341,-1231){\makebox(0,0)[b]{\smash{{\SetFigFont{10}{12.0}{\rmdefault}{\mddefault}{\updefault}{\color[rgb]{0,0,0}$AInit$}%
}}}}
\put(5131,-1456){\makebox(0,0)[b]{\smash{{\SetFigFont{10}{12.0}{\rmdefault}{\mddefault}{\updefault}{\color[rgb]{0,0,0}$ref$}%
}}}}
\end{picture}%

%% file: intv-predrel.pspdftex
\begin{picture}(0,0)%
\includegraphics{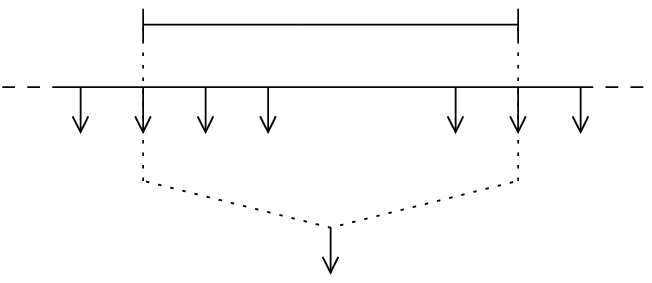}%
\end{picture}%
\setlength{\unitlength}{3947sp}%
\begingroup\makeatletter\ifx\SetFigFont\undefined%
\gdef\SetFigFont#1#2#3#4#5{%
  \reset@font\fontsize{#1}{#2pt}%
  \fontfamily{#3}\fontseries{#4}\fontshape{#5}%
  \selectfont}%
\fi\endgroup%
\begin{picture}(3177,1636)(139,-2150)
\put(526,-1411){\makebox(0,0)[b]{\smash{{\SetFigFont{12}{14.4}{\rmdefault}{\mddefault}{\updefault}{\color[rgb]{0,0,0}$\Sigma_Z$}%
}}}}
\put(826,-1411){\makebox(0,0)[b]{\smash{{\SetFigFont{12}{14.4}{\rmdefault}{\mddefault}{\updefault}{\color[rgb]{0,0,0}$\Sigma_Z$}%
}}}}
\put(1126,-1411){\makebox(0,0)[b]{\smash{{\SetFigFont{12}{14.4}{\rmdefault}{\mddefault}{\updefault}{\color[rgb]{0,0,0}$\Sigma_Z$}%
}}}}
\put(1426,-1411){\makebox(0,0)[b]{\smash{{\SetFigFont{12}{14.4}{\rmdefault}{\mddefault}{\updefault}{\color[rgb]{0,0,0}$\Sigma_Z$}%
}}}}
\put(2326,-1411){\makebox(0,0)[b]{\smash{{\SetFigFont{12}{14.4}{\rmdefault}{\mddefault}{\updefault}{\color[rgb]{0,0,0}$\Sigma_Z$}%
}}}}
\put(2626,-1411){\makebox(0,0)[b]{\smash{{\SetFigFont{12}{14.4}{\rmdefault}{\mddefault}{\updefault}{\color[rgb]{0,0,0}$\Sigma_Z$}%
}}}}
\put(2926,-1411){\makebox(0,0)[b]{\smash{{\SetFigFont{12}{14.4}{\rmdefault}{\mddefault}{\updefault}{\color[rgb]{0,0,0}$\Sigma_Z$}%
}}}}
\put(1726,-2086){\makebox(0,0)[b]{\smash{{\SetFigFont{12}{14.4}{\rmdefault}{\mddefault}{\updefault}{\color[rgb]{0,0,0}$\bool$}%
}}}}
\put(3301,-1036){\makebox(0,0)[lb]{\smash{{\SetFigFont{12}{14.4}{\rmdefault}{\mddefault}{\updefault}{\color[rgb]{0,0,0}$z$}%
}}}}
\put(1726,-661){\makebox(0,0)[b]{\smash{{\SetFigFont{12}{14.4}{\rmdefault}{\mddefault}{\updefault}{\color[rgb]{0,0,0}$\Delta$}%
}}}}
\put(1876,-1186){\rotatebox{360.0}{\makebox(0,0)[b]{\smash{{\SetFigFont{12}{14.4}{\rmdefault}{\mddefault}{\updefault}{\color[rgb]{0,0,0}. . . . .}%
}}}}}
\end{picture}%

%% file: intv-rel.pspdftex
\begin{picture}(0,0)%
\includegraphics{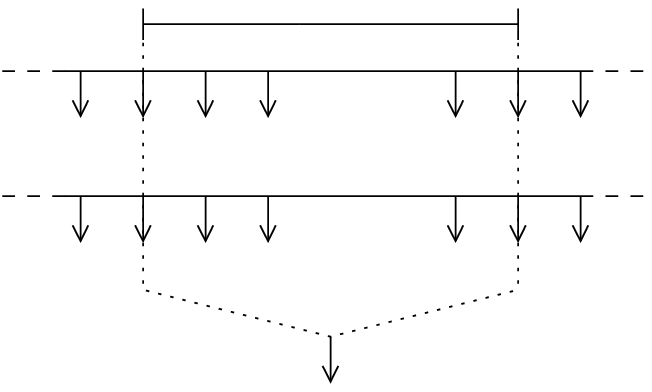}%
\end{picture}%
\setlength{\unitlength}{3947sp}%
\begingroup\makeatletter\ifx\SetFigFont\undefined%
\gdef\SetFigFont#1#2#3#4#5{%
  \reset@font\fontsize{#1}{#2pt}%
  \fontfamily{#3}\fontseries{#4}\fontshape{#5}%
  \selectfont}%
\fi\endgroup%
\begin{picture}(3177,2161)(139,-2150)
\put(1726,-136){\makebox(0,0)[b]{\smash{{\SetFigFont{12}{14.4}{\rmdefault}{\mddefault}{\updefault}{\color[rgb]{0,0,0}$\Delta$}%
}}}}
\put(526,-811){\makebox(0,0)[b]{\smash{{\SetFigFont{12}{14.4}{\rmdefault}{\mddefault}{\updefault}{\color[rgb]{0,0,0}$\Sigma_Y$}%
}}}}
\put(826,-811){\makebox(0,0)[b]{\smash{{\SetFigFont{12}{14.4}{\rmdefault}{\mddefault}{\updefault}{\color[rgb]{0,0,0}$\Sigma_Y$}%
}}}}
\put(1126,-811){\makebox(0,0)[b]{\smash{{\SetFigFont{12}{14.4}{\rmdefault}{\mddefault}{\updefault}{\color[rgb]{0,0,0}$\Sigma_Y$}%
}}}}
\put(1426,-811){\makebox(0,0)[b]{\smash{{\SetFigFont{12}{14.4}{\rmdefault}{\mddefault}{\updefault}{\color[rgb]{0,0,0}$\Sigma_Y$}%
}}}}
\put(2326,-811){\makebox(0,0)[b]{\smash{{\SetFigFont{12}{14.4}{\rmdefault}{\mddefault}{\updefault}{\color[rgb]{0,0,0}$\Sigma_Y$}%
}}}}
\put(2626,-811){\makebox(0,0)[b]{\smash{{\SetFigFont{12}{14.4}{\rmdefault}{\mddefault}{\updefault}{\color[rgb]{0,0,0}$\Sigma_Y$}%
}}}}
\put(2926,-811){\makebox(0,0)[b]{\smash{{\SetFigFont{12}{14.4}{\rmdefault}{\mddefault}{\updefault}{\color[rgb]{0,0,0}$\Sigma_Y$}%
}}}}
\put(526,-1411){\makebox(0,0)[b]{\smash{{\SetFigFont{12}{14.4}{\rmdefault}{\mddefault}{\updefault}{\color[rgb]{0,0,0}$\Sigma_Z$}%
}}}}
\put(826,-1411){\makebox(0,0)[b]{\smash{{\SetFigFont{12}{14.4}{\rmdefault}{\mddefault}{\updefault}{\color[rgb]{0,0,0}$\Sigma_Z$}%
}}}}
\put(1126,-1411){\makebox(0,0)[b]{\smash{{\SetFigFont{12}{14.4}{\rmdefault}{\mddefault}{\updefault}{\color[rgb]{0,0,0}$\Sigma_Z$}%
}}}}
\put(1426,-1411){\makebox(0,0)[b]{\smash{{\SetFigFont{12}{14.4}{\rmdefault}{\mddefault}{\updefault}{\color[rgb]{0,0,0}$\Sigma_Z$}%
}}}}
\put(2326,-1411){\makebox(0,0)[b]{\smash{{\SetFigFont{12}{14.4}{\rmdefault}{\mddefault}{\updefault}{\color[rgb]{0,0,0}$\Sigma_Z$}%
}}}}
\put(2626,-1411){\makebox(0,0)[b]{\smash{{\SetFigFont{12}{14.4}{\rmdefault}{\mddefault}{\updefault}{\color[rgb]{0,0,0}$\Sigma_Z$}%
}}}}
\put(2926,-1411){\makebox(0,0)[b]{\smash{{\SetFigFont{12}{14.4}{\rmdefault}{\mddefault}{\updefault}{\color[rgb]{0,0,0}$\Sigma_Z$}%
}}}}
\put(1726,-2086){\makebox(0,0)[b]{\smash{{\SetFigFont{12}{14.4}{\rmdefault}{\mddefault}{\updefault}{\color[rgb]{0,0,0}$\bool$}%
}}}}
\put(1876,-586){\rotatebox{360.0}{\makebox(0,0)[b]{\smash{{\SetFigFont{12}{14.4}{\rmdefault}{\mddefault}{\updefault}{\color[rgb]{0,0,0}. . . . .}%
}}}}}
\put(1876,-1186){\rotatebox{360.0}{\makebox(0,0)[b]{\smash{{\SetFigFont{12}{14.4}{\rmdefault}{\mddefault}{\updefault}{\color[rgb]{0,0,0}. . . . .}%
}}}}}
\put(3301,-1036){\makebox(0,0)[lb]{\smash{{\SetFigFont{12}{14.4}{\rmdefault}{\mddefault}{\updefault}{\color[rgb]{0,0,0}$z$}%
}}}}
\put(3301,-436){\makebox(0,0)[lb]{\smash{{\SetFigFont{12}{14.4}{\rmdefault}{\mddefault}{\updefault}{\color[rgb]{0,0,0}$y$}%
}}}}
\end{picture}%

%% file: intvfs.pspdftex
\begin{picture}(0,0)%
\includegraphics{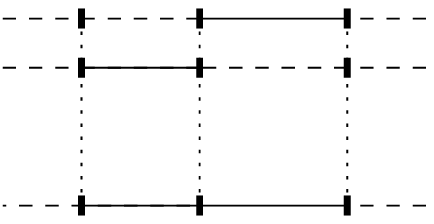}%
\end{picture}%
\setlength{\unitlength}{4144sp}%
\begingroup\makeatletter\ifx\SetFigFont\undefined%
\gdef\SetFigFont#1#2#3#4#5{%
  \reset@font\fontsize{#1}{#2pt}%
  \fontfamily{#3}\fontseries{#4}\fontshape{#5}%
  \selectfont}%
\fi\endgroup%
\begin{picture}(2052,1225)(2329,-1916)
\put(3601,-1456){\makebox(0,0)[b]{\smash{{\SetFigFont{10}{12.0}{\rmdefault}{\mddefault}{\updefault}{\color[rgb]{0,0,0}$ref$}%
}}}}
\put(2971,-1456){\makebox(0,0)[b]{\smash{{\SetFigFont{10}{12.0}{\rmdefault}{\mddefault}{\updefault}{\color[rgb]{0,0,0}$ref$}%
}}}}
\put(2971,-1681){\makebox(0,0)[b]{\smash{{\SetFigFont{10}{12.0}{\rmdefault}{\mddefault}{\updefault}{\color[rgb]{0,0,0}$\Delta_0$}%
}}}}
\put(3601,-1681){\makebox(0,0)[b]{\smash{{\SetFigFont{10}{12.0}{\rmdefault}{\mddefault}{\updefault}{\color[rgb]{0,0,0}$\Delta$}%
}}}}
\put(3601,-1861){\makebox(0,0)[b]{\smash{{\SetFigFont{10}{12.0}{\rmdefault}{\mddefault}{\updefault}{\color[rgb]{0,0,0}$h$}%
}}}}
\put(3601,-826){\makebox(0,0)[b]{\smash{{\SetFigFont{10}{12.0}{\rmdefault}{\mddefault}{\updefault}{\color[rgb]{0,0,0}$g$}%
}}}}
\put(4366,-1771){\makebox(0,0)[lb]{\smash{{\SetFigFont{10}{12.0}{\rmdefault}{\mddefault}{\updefault}{\color[rgb]{0,0,0}$z$}%
}}}}
\put(4366,-1141){\makebox(0,0)[lb]{\smash{{\SetFigFont{10}{12.0}{\rmdefault}{\mddefault}{\updefault}{\color[rgb]{0,0,0}$y_0$}%
}}}}
\put(4366,-916){\makebox(0,0)[lb]{\smash{{\SetFigFont{10}{12.0}{\rmdefault}{\mddefault}{\updefault}{\color[rgb]{0,0,0}$y$}%
}}}}
\end{picture}%

%% file: soundnessa.pspdftex
\begin{picture}(0,0)%
\includegraphics{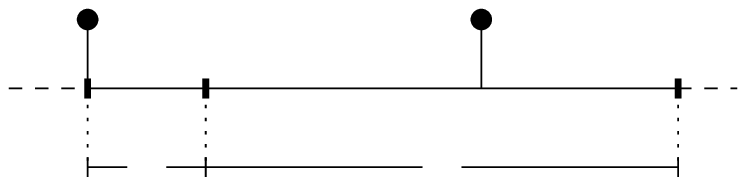}%
\end{picture}%
\setlength{\unitlength}{4144sp}%
\begingroup\makeatletter\ifx\SetFigFont\undefined%
\gdef\SetFigFont#1#2#3#4#5{%
  \reset@font\fontsize{#1}{#2pt}%
  \fontfamily{#3}\fontseries{#4}\fontshape{#5}%
  \selectfont}%
\fi\endgroup%
\begin{picture}(3447,874)(2236,-485)
\put(4321,-421){\makebox(0,0)[b]{\smash{{\SetFigFont{10}{12.0}{\rmdefault}{\mddefault}{\updefault}{\color[rgb]{0,0,0}$\Delta$}%
}}}}
\put(4636, 74){\makebox(0,0)[b]{\smash{{\SetFigFont{10}{12.0}{\rmdefault}{\mddefault}{\updefault}{\color[rgb]{0,0,0}$CF$}%
}}}}
\put(2971,-196){\makebox(0,0)[b]{\smash{{\SetFigFont{10}{12.0}{\rmdefault}{\mddefault}{\updefault}{\color[rgb]{0,0,0}$CI$}%
}}}}
\put(4681,254){\makebox(0,0)[b]{\smash{{\SetFigFont{10}{12.0}{\rmdefault}{\mddefault}{\updefault}{\color[rgb]{0,0,0}$\sigma'$}%
}}}}
\put(2566,254){\makebox(0,0)[b]{\smash{{\SetFigFont{10}{12.0}{\rmdefault}{\mddefault}{\updefault}{\color[rgb]{0,0,0}$\sigma$}%
}}}}
\put(2251,-61){\makebox(0,0)[b]{\smash{{\SetFigFont{10}{12.0}{\rmdefault}{\mddefault}{\updefault}{\color[rgb]{0,0,0}$z$}%
}}}}
\put(4276,-196){\makebox(0,0)[b]{\smash{{\SetFigFont{10}{12.0}{\rmdefault}{\mddefault}{\updefault}{\color[rgb]{0,0,0}$\bigwedge_{p:P} COp_p$}%
}}}}
\put(2971,-421){\makebox(0,0)[b]{\smash{{\SetFigFont{10}{12.0}{\rmdefault}{\mddefault}{\updefault}{\color[rgb]{0,0,0}$\Delta_0$}%
}}}}
\end{picture}%

%% file: soundnessb.pspdftex
\begin{picture}(0,0)%
\includegraphics{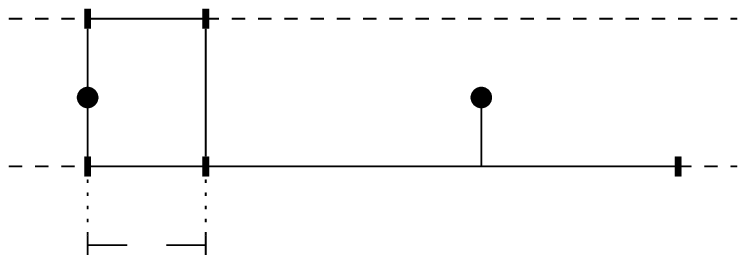}%
\end{picture}%
\setlength{\unitlength}{4144sp}%
\begingroup\makeatletter\ifx\SetFigFont\undefined%
\gdef\SetFigFont#1#2#3#4#5{%
  \reset@font\fontsize{#1}{#2pt}%
  \fontfamily{#3}\fontseries{#4}\fontshape{#5}%
  \selectfont}%
\fi\endgroup%
\begin{picture}(3447,1312)(2236,-485)
\put(2971,-421){\makebox(0,0)[b]{\smash{{\SetFigFont{10}{12.0}{\rmdefault}{\mddefault}{\updefault}{\color[rgb]{0,0,0}$\Delta_0$}%
}}}}
\put(4636, 74){\makebox(0,0)[b]{\smash{{\SetFigFont{10}{12.0}{\rmdefault}{\mddefault}{\updefault}{\color[rgb]{0,0,0}$CF$}%
}}}}
\put(2971,-196){\makebox(0,0)[b]{\smash{{\SetFigFont{10}{12.0}{\rmdefault}{\mddefault}{\updefault}{\color[rgb]{0,0,0}$CI$}%
}}}}
\put(4681,254){\makebox(0,0)[b]{\smash{{\SetFigFont{10}{12.0}{\rmdefault}{\mddefault}{\updefault}{\color[rgb]{0,0,0}$\sigma'$}%
}}}}
\put(2566,254){\makebox(0,0)[b]{\smash{{\SetFigFont{10}{12.0}{\rmdefault}{\mddefault}{\updefault}{\color[rgb]{0,0,0}$\sigma$}%
}}}}
\put(2971,254){\makebox(0,0)[b]{\smash{{\SetFigFont{10}{12.0}{\rmdefault}{\mddefault}{\updefault}{\color[rgb]{0,0,0}$ref$}%
}}}}
\put(2971,704){\makebox(0,0)[b]{\smash{{\SetFigFont{10}{12.0}{\rmdefault}{\mddefault}{\updefault}{\color[rgb]{0,0,0}$AI$}%
}}}}
\put(2251,614){\makebox(0,0)[b]{\smash{{\SetFigFont{10}{12.0}{\rmdefault}{\mddefault}{\updefault}{\color[rgb]{0,0,0}$y_0$}%
}}}}
\put(2251,-61){\makebox(0,0)[b]{\smash{{\SetFigFont{10}{12.0}{\rmdefault}{\mddefault}{\updefault}{\color[rgb]{0,0,0}$z$}%
}}}}
\put(4276,-196){\makebox(0,0)[b]{\smash{{\SetFigFont{10}{12.0}{\rmdefault}{\mddefault}{\updefault}{\color[rgb]{0,0,0}$\bigwedge_{p:P} COp_p$}%
}}}}
\end{picture}%

%% file: soundnessc.pspdftex
\begin{picture}(0,0)%
\includegraphics{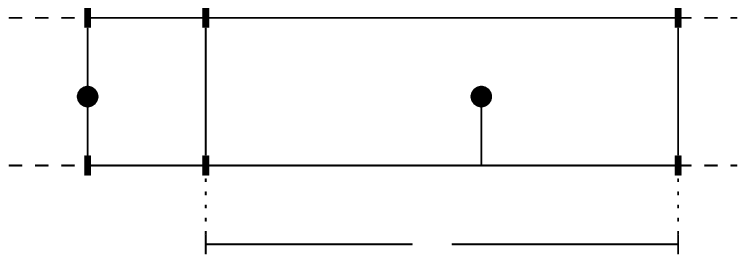}%
\end{picture}%
\setlength{\unitlength}{4144sp}%
\begingroup\makeatletter\ifx\SetFigFont\undefined%
\gdef\SetFigFont#1#2#3#4#5{%
  \reset@font\fontsize{#1}{#2pt}%
  \fontfamily{#3}\fontseries{#4}\fontshape{#5}%
  \selectfont}%
\fi\endgroup%
\begin{picture}(3447,1372)(2236,-476)
\put(4276,-421){\makebox(0,0)[b]{\smash{{\SetFigFont{10}{12.0}{\rmdefault}{\mddefault}{\updefault}{\color[rgb]{0,0,0}$\Delta$}%
}}}}
\put(4636, 74){\makebox(0,0)[b]{\smash{{\SetFigFont{10}{12.0}{\rmdefault}{\mddefault}{\updefault}{\color[rgb]{0,0,0}$CF$}%
}}}}
\put(2971,-196){\makebox(0,0)[b]{\smash{{\SetFigFont{10}{12.0}{\rmdefault}{\mddefault}{\updefault}{\color[rgb]{0,0,0}$CI$}%
}}}}
\put(4681,254){\makebox(0,0)[b]{\smash{{\SetFigFont{10}{12.0}{\rmdefault}{\mddefault}{\updefault}{\color[rgb]{0,0,0}$\sigma'$}%
}}}}
\put(2566,254){\makebox(0,0)[b]{\smash{{\SetFigFont{10}{12.0}{\rmdefault}{\mddefault}{\updefault}{\color[rgb]{0,0,0}$\sigma$}%
}}}}
\put(2971,254){\makebox(0,0)[b]{\smash{{\SetFigFont{10}{12.0}{\rmdefault}{\mddefault}{\updefault}{\color[rgb]{0,0,0}$ref$}%
}}}}
\put(2971,704){\makebox(0,0)[b]{\smash{{\SetFigFont{10}{12.0}{\rmdefault}{\mddefault}{\updefault}{\color[rgb]{0,0,0}$AI$}%
}}}}
\put(2251,614){\makebox(0,0)[b]{\smash{{\SetFigFont{10}{12.0}{\rmdefault}{\mddefault}{\updefault}{\color[rgb]{0,0,0}$y$}%
}}}}
\put(2251,-61){\makebox(0,0)[b]{\smash{{\SetFigFont{10}{12.0}{\rmdefault}{\mddefault}{\updefault}{\color[rgb]{0,0,0}$z$}%
}}}}
\put(4141,254){\makebox(0,0)[b]{\smash{{\SetFigFont{10}{12.0}{\rmdefault}{\mddefault}{\updefault}{\color[rgb]{0,0,0}$ref$}%
}}}}
\put(4276,749){\makebox(0,0)[b]{\smash{{\SetFigFont{10}{12.0}{\rmdefault}{\mddefault}{\updefault}{\color[rgb]{0,0,0}$\bigwedge_{p:P} AOp_p$}%
}}}}
\put(4276,-196){\makebox(0,0)[b]{\smash{{\SetFigFont{10}{12.0}{\rmdefault}{\mddefault}{\updefault}{\color[rgb]{0,0,0}$\bigwedge_{p:P} COp_p$}%
}}}}
\end{picture}%

%% file: soundnesse.pspdftex
\begin{picture}(0,0)%
\includegraphics{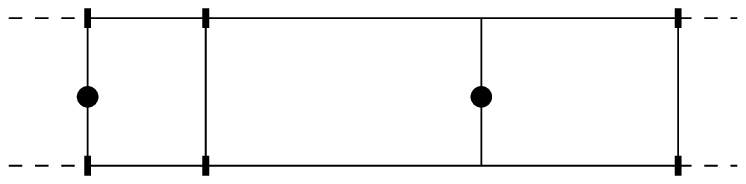}%
\end{picture}%
\setlength{\unitlength}{4144sp}%
\begingroup\makeatletter\ifx\SetFigFont\undefined%
\gdef\SetFigFont#1#2#3#4#5{%
  \reset@font\fontsize{#1}{#2pt}%
  \fontfamily{#3}\fontseries{#4}\fontshape{#5}%
  \selectfont}%
\fi\endgroup%
\begin{picture}(3447,1156)(2236,-260)
\put(4636,434){\makebox(0,0)[b]{\smash{{\SetFigFont{10}{12.0}{\rmdefault}{\mddefault}{\updefault}{\color[rgb]{0,0,0}$AF$}%
}}}}
\put(4636, 74){\makebox(0,0)[b]{\smash{{\SetFigFont{10}{12.0}{\rmdefault}{\mddefault}{\updefault}{\color[rgb]{0,0,0}$CF$}%
}}}}
\put(2971,-196){\makebox(0,0)[b]{\smash{{\SetFigFont{10}{12.0}{\rmdefault}{\mddefault}{\updefault}{\color[rgb]{0,0,0}$CI$}%
}}}}
\put(4681,254){\makebox(0,0)[b]{\smash{{\SetFigFont{10}{12.0}{\rmdefault}{\mddefault}{\updefault}{\color[rgb]{0,0,0}$\sigma'$}%
}}}}
\put(2566,254){\makebox(0,0)[b]{\smash{{\SetFigFont{10}{12.0}{\rmdefault}{\mddefault}{\updefault}{\color[rgb]{0,0,0}$\sigma$}%
}}}}
\put(2971,254){\makebox(0,0)[b]{\smash{{\SetFigFont{10}{12.0}{\rmdefault}{\mddefault}{\updefault}{\color[rgb]{0,0,0}$ref$}%
}}}}
\put(2971,704){\makebox(0,0)[b]{\smash{{\SetFigFont{10}{12.0}{\rmdefault}{\mddefault}{\updefault}{\color[rgb]{0,0,0}$AI$}%
}}}}
\put(2251,614){\makebox(0,0)[b]{\smash{{\SetFigFont{10}{12.0}{\rmdefault}{\mddefault}{\updefault}{\color[rgb]{0,0,0}$y$}%
}}}}
\put(2251,-61){\makebox(0,0)[b]{\smash{{\SetFigFont{10}{12.0}{\rmdefault}{\mddefault}{\updefault}{\color[rgb]{0,0,0}$z$}%
}}}}
\put(4276,749){\makebox(0,0)[b]{\smash{{\SetFigFont{10}{12.0}{\rmdefault}{\mddefault}{\updefault}{\color[rgb]{0,0,0}$\bigwedge_{p:P} AOp_p$}%
}}}}
\put(4276,-196){\makebox(0,0)[b]{\smash{{\SetFigFont{10}{12.0}{\rmdefault}{\mddefault}{\updefault}{\color[rgb]{0,0,0}$\bigwedge_{p:P} COp_p$}%
}}}}
\put(3871,254){\makebox(0,0)[b]{\smash{{\SetFigFont{10}{12.0}{\rmdefault}{\mddefault}{\updefault}{\color[rgb]{0,0,0}$ref$}%
}}}}
\end{picture}%